%% file: perfect_crystal_rHF_v9.tex
\documentclass[a4paper,10pt]{article}

\input{usepackages}
\input{list_commands}

\usepackage[titletoc]{appendix}

\newcommand{\cutoff}{{\mathrm{cutoff}}}

\begin{document}

\title{Convergence rates of supercell calculations in the reduced Hartree-Fock model}

\author{
  David Gontier, Salma Lahbabi
}

\maketitle

\begin{abstract}
This article is concerned with the numerical simulations of perfect crystals. We study the rate of convergence of the reduced Hartree-Fock (rHF) model in a supercell towards the periodic rHF model in the whole space. We prove that, whenever the crystal is an insulator or a semi-conductor, the supercell energy per unit cell converges exponentially fast towards the periodic rHF energy per unit cell, with respect to the size of the supercell.
\end{abstract}


\section{Introduction}

The numerical simulation of the electronic structure of crystals is a very active research area in solid state physics, materials science and nano-electronics. When the crystal is perfect, a good approximation of its electronic ground state density can be obtained by solving a mean-field nonlinear periodic model set on the whole space. 
Using the Bloch transform~\cite[Chapter XIII]{ReedSimon4}, we can recast such a problem as a continuous family of compact problems indexed by points of the Brillouin-zone. In practice, the compact problems are solved on a discretization of the Brillouin-zone.
There is therefore an inherent error coming from the fact that the Brillouin-zone is sampled, and it is not obvious \textit{a priori} whether this error is small, due to the nonlinearity of the problem. It has been observed numerically since the work of Monkhorst and Pack~\cite{Monkhorst1976} that this error is indeed very small when the discretization is uniform, and when the crystal is an insulator or a semiconductor. To our knowledge, no rigorous proof of this fact was ever given. This article aims at proving why it is indeed the case in the reduced Hartree-Fock (rHF) model, which is a Hartree-Fock model in which the exchange term is neglected. This model was studied in~\cite{Cances2008, Catto2001}. \\

A crystal is modeled by a periodic nuclear charge distribution $\mu_\per$. The corresponding rHF energy per unit cell is denoted by $I_\per^{\mu_\per}$. When numerical calculations are performed over a regular discretization of the Brillouin-zone, this amounts to calculate the energy on a \textit{supercell}, \textit{i.e.} on a large box containing $L$ times the periodicity of $\mu_\per$ in each direction (for a total of $L^3$ unit cells in the supercell), and with periodic boundary conditions. The rHF energy on a supercell of size~$L$ is denoted by $I^{\mu_\per}_L$, so that the corresponding energy per unit cell is $L^{-3} I^{\mu_\per}_L$.\\

It was proved in~\cite{Cances2008} that $L^{-3} I^{\mu_\per}_L$ converges to $I_\per^{\mu_\per}$ as $L$ goes to infinity, when the crystal is an insulator or a semiconductor. However, following the proof in~\cite{Cances2008}, we find a rate of convergence of order $L^{-1}$, which is well below what is numerically observed. Our main result is that, if the crystal is an insulator or a semiconductor, then there exist constants $C \in \R^+$ and $\alpha > 0$, such that
\begin{equation} \label{eq:intro_exp_CV}
	\forall L \in \N^*, \quad \left| L^{-3} I^{\mu_\per}_L - I_\per^{\mu_\per} \right| \le C \re^{- \alpha L}.
\end{equation}
We also prove that the supercell electronic density converges exponentially fast to the periodic rHF electronic density, in the $L^\infty(\R^3)$ norm. To prove such rates of convergence, we recast the problem into the difference between an integral and a corresponding Riemann sum, and show that the integrand is both periodic and analytic on a complex strip. Similar tools were used in~\cite{Cloizeaux1964, Cloizeaux1964a, Kohn1959, Brouder2007, Panati2007} to prove that the Wannier functions of insulators are exponentially localized.\\

This article is organized as follows. In Section~\ref{sec:rHF}, we recall how the rHF model is derived, and present the main results. In Section~\ref{sec:Bloch}, we apply the Bloch theory for both periodic models and supercell models. The proofs of the main results are postponed until Section~\ref{sec:proof_expCV}. Finally, we illustrate our theoretical results with numerical simulations in Section~\ref{sec:numerics_rHF}.\\

Throughout this article, we will give explicit values of the constants appearing in the inequalities. These values are very crude, but allows one to see how these constants depend on the parameters of the electronic problem.

%
%
\section{Presentation of the models}
\label{sec:rHF}

A perfect crystal is a periodic arrangement of atoms. Both the nuclear charge density $\mu_\per$ and the electronic density are $\mathcal R$-periodic functions, where $\Lat$ is a discrete periodic lattice of $\R^3$. Let $\WS$ be the Wigner-Seitz cell of the lattice, and let $\BZ$ be the Wigner-Seitz cell of the dual lattice $\RLat$ (one can also take $\BZ$ as the first Brillouin-zone). For instance, for $\Lat= a \ZZ^3$, $\Gamma= [-a/2,a/2)^3$, $\RLat= (2\pi/a) \ZZ^3$ and $\Gamma^\ast=[-\pi/a,\pi/a)^3$. For $\bR \in \Lat$, we let $\tau_\bR$ be the translation operator on~$L^2(\R^3)$ defined by $(\tau_\bR f)( \bx ) := f(\bx - \bR)$. \\

We will assume throughout the paper that the nuclear charge density $\mu_\per$ is in $L^2_\per(\WS)$ for simplicity, but distributions with singularity points may also be handled~\cite{Blanc2003}.


\subsection{The supercell rHF model}

In a supercell model, the system is confined to a box $\WS_L := L \WS$ with periodic boundary conditions. We denote by $L^2_\per(\WS_L)$ the Hilbert space of locally square integrable functions that are $L\Lat$-periodic. The Fourier coefficients of a function $f \in L^2_\per ( \WS_L)$ are defined by
\[
	\forall \bk \in L^{-1} \RLat , \quad c_\bk^L (f) =\frac{1}{| \Gamma_L|} \int_{\Gamma_L} f(\bx) \re^{-\ri \bk \cdot \bx} \rd \bx,
\]
so that, for any $f \in L^2_\per(\Gamma_L)$,
\[
	f(\bx) = \sum_{\bk \in L^{-1} \RLat} c_\bk^L(f) \re^{\ri \bk \cdot \bx} \quad \text{ a.e. and in } L^2_\per(\WS_L).
\]

The set of admissible electronic states for the supercell model is
\[
	\cP_L := \left\{ \gamma_L \in \cS(L^2_\per(\WS_L)), \ 0 \le \gamma_L \le 1,  \  \Tr_{L^2_\per(\WS_L)} \left(  \gamma_L \right) +\Tr_{L^2_\per(\WS_L)}\left(  -\Delta_L \gamma_L \right) < \infty \right\},
\]
where $\cS(\cH)$ denotes the space of the bounded self-adjoint operators on the Hilbert space $\cH$. Here, $\Tr_{L^2_\per(\WS_L)}\left(  -\Delta_L \gamma_L \right)$ is a shorthand notation for 
\begin{equation} \label{eq:def:Pj}
	\Tr_{L^2_\per(\WS_L)}\left( - \Delta_L \gamma_L \right):=\sum_{i=1}^3\Tr_{L^2_\per(\WS_L)}\left(  P_{j,L} \gamma_L P_{j,L} \right),
\end{equation}
where, for $1 \le j \le 3$,  $P_{j,L}$ is the self-adjoint operator on $L^2_\per (\Gamma_L)$ defined by $c_\bk^L(P_{j,L} f)= k_j c_\bk^L(f)$ for all $\bk = (k_1, k_2, k_3) \in L^{-1} \RLat$. Note that $c_\bk^L(-\Delta_L f)= | \bk |^2 c_\bk^L(f)$ for all $\bk \in L^{-1} \RLat $. 

We introduce the $L \Lat$-periodic Green kernel $G_L$ of the Poisson interaction~\cite{Lieb1977}, solution of
\[
	\left\{ \begin{array}{c}
		-\Delta G_L =  4 \pi \left( \sum_{\bk \in \Lat} \delta_\bk - 1 \right) \\
		G_L \text{ is } L\Lat\text{-periodic.}
		\end{array} \right. 
\]
The expression of $G_L$ is given in the Fourier basis by
\begin{equation} \label{eq:GLFourier}
	G_L (\bx) = c_L + \frac{4 \pi}{| \WS_L|} \sum_{\bk \in L^{-1} \RLat \setminus \{ \bnull \}} \dfrac{\re^{\ri \bk \cdot \bx}}{| \bk |^2},
\end{equation}
where $c_L = | \WS_L|^{-1} \int_{\WS_L} G_L$. The constant $c_L$ can be any fixed constant \textit{a priori}. In one of the first article on the topic~\cite{Lieb1977}, the authors chose to set $c_L = 0$, but other choices are equally valid (see~\cite{Cances2008} for instance). This is due to the fact that $c_L$ does not play any role for neutral systems.  We choose to set $c_L = 0$ for simplicity. 
The supercell Coulomb energy is defined by
\begin{equation} \label{eq:def_D1}
	\forall f, g \in L^2_\per(\WS_L), \quad D_L (f,g) := \int_\WS (f\ast_{\Gamma_L} G_L) (\bx) g(\bx) \rd \bx,
\end{equation}
where $(f\ast_{\Gamma_L} G_L) (\bx):= \int_\WS f(\by)G_L(\bx-\by) \rd \by$. We recall that the map $\rho \mapsto \rho \ast_{\Gamma_L} G_L$ is continuous from $L^2_\per(\WS_L)$ to $L^\infty_\per(\WS_L)$.\\

Any $\gamma_L \in \cP_L$ is locally trace-class, and can be associated an $L \Lat$-periodic density $\rho_{\gamma_L} \in L^2_\per(\WS_L)$.
For $\gamma_L \in \cP_L$, the supercell reduced Hartree-Fock energy is
\begin{equation} \label{eq:EL}
	\cE_L^{\mu_\per}(\gamma_L) :=  \dfrac12  \Tr_{L^2_\per(\WS_L)} (- \Delta_L \gamma_L) + \dfrac12 D_L(\rho_{\gamma_L} - \mu_\per, \rho_{\gamma_L} - \mu_\per).
\end{equation}
The first term of~\eqref{eq:EL} corresponds to the supercell kinetic energy, and the second term represents the supercell Coulomb energy.
The ground state energy of the system is given by the minimization problem
\begin{equation} \label{eq:finite_problem}
	I^{\mu_\per}_L = \inf \left\{ \cE^{\mu_\per}_L (\gamma_L) , \ \gamma_L \in \cP_L, \  \int_{\WS_L} \rho_{\gamma_L} = \int_{\WS_L} \mu_\per   \right\}.
\end{equation}
Using techniques similar to~\cite[Theorem 4]{Cances2008}, the following result holds (we do not prove it, for the arguments are similar to the ones in~\cite{Cances2008}).
\begin{theorem}[Existence of a supercell minimizer] \label{th:existence_supercell}
For all $L \in \N^*$, the minimization problem~\eqref{eq:finite_problem} admits minimizers. One of these minimizers $\gamma_{L,0}$ satisfies $\tau_\bR \gamma_{L,0} = \gamma_{L,0} \tau_\bR$. All minimizers share the same density $\rho_{\gamma_{L,0}}$, which is $\Lat$-periodic.
 Finally, $\gamma_{L,0}$ satisfies the self-consistent equation
\begin{equation} \label{eq:sc_supercell}
	\left\{ \begin{array}{ll}
		\gamma_{L,0} & = \mathds{1} \left( H_{L,0} < \varepsilon_F^L \right) + \delta \\
		H_{L,0} & = - \frac12 \Delta_L + V_{L,0} \\
		V_{L,0} & = \left( \rho_{\gamma_{L,0}} - \mu_\per \right) \ast_\WS G_1.
	\end{array} \right.
\end{equation}
where $H_{L,0}$ acts on $L^2_\per(\WS_L)$ and $0 \le \delta \le \mathds{1}(H_{L,0} = \varepsilon_F^L)$ is a finite rank operator.
\end{theorem}
Here, $\varepsilon_F^L$ is the Fermi level of the supercell model. It is chosen so that the charge constraint in~\eqref{eq:finite_problem} is satisfied.

\begin{remark}
	The $L \Lat$-periodic density of the minimizers $\rho_{\gamma_{L,0}}$ is actually $\Lat$-periodic. It is unclear that such a property should hold for more complex models (\textit{e.g.} Kohn-Sham models). This is the reason why we state our results for the rHF model. We believe however that similar results should hold true for more complex systems, provided that the supercell density is $\Lat$-periodic for each size of the supercell.
\end{remark}

%
%

\subsection{The reduced Hartree-Fock model for perfect crystals}

The rHF model for perfect crystals, or periodic rHF, has been rigorously derived from the rHF model for finite molecular systems by means of a thermodynamic limit procedure by Catto, Le Bris and Lions~\cite{Catto2001}. In~\cite{Cances2008}, Cancès, Deleurence and Lewin proved that the same periodic rHF model is the limit of the rHF supercell model as the size of the supercell goes to infinity. \\

We introduce the set of admissible density matrices
\begin{equation} \label{eq:VTr1}
	\cP_{\per} := \left\{ \gamma \in \cS(L^2_\per(\Gamma)), \ 0 \le \gamma \le 1, \ \forall \bR \in \Lat, \ \tau_\bR \gamma = \gamma \tau_\bR , \ \VTr \left( \gamma \right) +\VTr\left( -\Delta\gamma \right)< \infty \right\},
\end{equation}
where $\VTr$ denotes the trace per unit volume. For any locally trace class operator $A$ that commutes with $\Lat$-translations, it reads
\begin{equation} \label{eq:VTr_def1}
\VTr\left(  A \right):= \lim_{L \to \infty} \dfrac{1}{L^3} \Tr\left( \1_{L\Gamma} A \1_{L\Gamma} \right).
\end{equation}
The trace per unit volume $\VTr$ can also be defined via the Bloch transform (see Equation~\eqref{eq:VTr} below). 
Here, $\VTr\left( -\Delta\gamma \right)$ is a shorthand notation for 
$$
	\VTr\left( -\Delta\gamma \right) := \sum_{j=1}^3\VTr\left( P_j\gamma P_j \right),
$$
where $P_j=-\ri\partial_{x_j}$ is the momentum operator in the $j^{\rm th}$ direction. The Coulomb energy per unit volume is defined by
\begin{equation} \label{eq:def_D1}
	\forall f, g \in L^2_\per(\WS), \quad D_1 (f,g) := \int_\WS (f\ast_\Gamma G_1) (\bx) g(\bx) \rd \bx,
\end{equation}
where $G_1$ was introduced in~\eqref{eq:GLFourier}.\\

Any $\gamma \in \cP_\per$ is locally trace-class, and can be associated an $\Lat$-periodic density $\rho_{\gamma} \in L^2_\per(\WS)$. For $\gamma \in \cP_\per$, the reduced Hartree-Fock energy is given by
\begin{equation} \label{eq:cE_per}
	\cE_\per^{\mu_\per}(\gamma) := \frac12 \VTr \left( -\Delta \gamma \right) + \dfrac12 D_1 \left( \rho_\gamma - \mu_\per, \rho_\gamma - \mu_\per \right).
\end{equation}
 The first term of~\eqref{eq:cE_per} corresponds to the kinetic energy per unit volume, and the second term represents the Coulomb energy per unit volume.
Finally, the periodic rHF ground state energy is given by the minimization problem
\begin{equation} \label{eq:Iper}
	I_\per^{\mu_\per} := \inf \left\{ \cE_\per^{\mu_\per}(\gamma), \ \gamma \in \cP_\per, \ \int_\WS \rho_\gamma = \int_\WS \mu_\per \right\}.
\end{equation}
It has been proved in \cite{Cances2008} that the minimization problem (\ref{eq:Iper}) admits a unique minimizer $\gamma_0$, which is the solution of the self-consistent equation
\begin{equation} \label{eq:sc}
	\left\{ \begin{array}{lll}
		\gamma_0 & = & \mathds{1}(H_0 < \varepsilon_F) + \delta \\
		H_{0} & = & - \frac12 \Delta + V_0 \\
		V_0 & = & (\rho_{\gamma_0} - \mu_\per) \ast_\WS G_1,
		\end{array} \right.
\end{equation}
where $H_0$ acts on $L^2(\R^3)$ and $0 \le \delta \le \mathds{1}(H_0 = \varepsilon_F)$ is a finite rank operator.
Here, the Fermi energy $\varepsilon_F$ is the Lagrange multiplier corresponding to the charge constraint $\int_\WS \rho_{\gamma_0} = \int_\WS \mu_\per$. We make the following assumption:\\
\[
	\boxed{ \text{\textbf{(A1)} The system is an insulator, in the sense that $H_0$ has a spectral gap around $\varepsilon_F$.}}
\]
In particular, $\delta = 0$.


\section{Main results}

Our main results are concerned with the rate of convergence of supercell models towards corresponding periodic models. We first prove the exponential rate of convergence in a linear setting, where the mean-filed potential $V$ is a fixed $\Lat$-periodic function: $V \in L^\infty_\per(\WS)$. We then extend our result to the nonlinear rHF model, where the external potential is the solution of the  self-consistent equation~\eqref{eq:sc_supercell} or~\eqref{eq:sc}.

We start with the linear case. The proof of the following proposition is given in Section~\ref{subsec:proof_exp_lin}.

\begin{proposition}[Convergence rate of the linear supercell model] \label{prop:expCV}
Let $V\in L^\infty_\per(\WS)$ be such that the operator $H=-\frac12 \Delta +V$ acting on $L^2(\R^3)$ has a gap of size $g > 0$ centered around the Fermi level $\varepsilon_F$. Then, for any $L \in \N^*$, the operator $H^L := - \frac12 \Delta_L + V$ acting on $L^2_\per(\WS_L)$ has a gap of size at least $g$ around $\varepsilon_F$. Let 
\begin{equation} \label{eq:gamma_gammaL}
	\gamma=\1\left(H \leq \varepsilon_F\right) \quad \text{and}\quad \gamma_L =\1\left(H^L \leq \varepsilon_F\right).
\end{equation}
	Then, $\gamma \in \cP_\per$ and $\gamma_L \in \cP_L$, and there exist constants $C \in \R^+$ and $\alpha >0$, that depend on the lattice $\Lat$, $\| V \|_{L^\infty}$, $g$ and $\varepsilon_F$ only, such that
	\begin{equation} \label{eq:conv_energy_lin}
		\forall L \in \N^*, \quad \left| \VTr \left( \gamma H \right) - \VTr_L \left( \gamma_L H^L \right) \right| \le C \re^{ - \alpha L} \quad \text{(ground state energy per unit volume)}
	\end{equation}
	and
	\begin{equation} \label{eq:conv_density}
		\forall L \in \N^*, \quad \left\| \rho_{\gamma} - \rho_{\gamma_L} \right\|_{L^\infty} \le C \re^{ - \alpha L} \quad \text{(ground state density)}.
	\end{equation}
\end{proposition}

In a second step, we will use the projectors $\gamma$ and $\gamma_L$ obtained for well chosen potentials $V$ as candidates for the minimization problems~\eqref{eq:Iper} and~\eqref{eq:finite_problem} respectively. We have the following result (see Section~\ref{sssec:proof_cor} for the proof).

\begin{corollary} \label{cor:expCV}
	With the same notation as in Proposition~\ref{prop:expCV}, there exist constants $C \in \R^+$ and $\alpha >0$, that depend on the lattice $\Lat$, $\| V \|_{L^\infty}$, $g$ and $\varepsilon_F$ only, such that
	\begin{equation} \label{eq:conv_energy}
		\forall L \in \N^*, \quad \left| \cE_\per^{\mu_\per}( \gamma) - L^{-3} \cE_L^{\mu_\per}(\gamma_L) \right| \le C \re^{- \alpha L}.
	\end{equation}
\end{corollary}

 
We are now able to state our main result for the rHF model. The proof of the following theorem is given in Section~\ref{subsec:proof_exp}. In the sequel, we denote by $\cB(E)$ the set of bounded operators acting on the Banach space $E$.

\begin{theorem}[Convergence rate of the rHF supercell model] \label{th:Deleurence_exp} $\,$
Under hypothesis ({\bf A1}), there exist $C \in \R^+$ and $\alpha  > 0$ independent of $L$ such that the following estimates hold true:
\begin{itemize}
	\item convergence of the ground state energy per unit volume:
	\[
		\forall L \in \N^*, \quad | L^{-3} I_L^{\mu_\per} - I_\per^{\mu_\per} | \le C \re^{- \alpha L} ;
	\]
	\item convergence of the ground state density:
	\[
		\forall L \in \N^*, \quad \| \rho_{\gamma_{L,0}} - \rho_{\gamma_0} \|_{L^\infty_\per(\WS)} \le C \re^{- \alpha L} ;
	\]
	\item convergence of the mean-field Hamiltonian:
	\[
		\forall L \in \N^*, \quad \| H_L - H_0 \|_{\cB(L^2(\R^3))}  \le C \re^{- \alpha L},
	\]
\end{itemize}
where $H_L := - \frac12 \Delta + \left( \rho_{\gamma_{L,0}} - \mu_\per \right) \ast_\WS G_1$ and $H := - \frac12 \Delta + \left( \rho_{\gamma_{0}} - \mu_\per \right) \ast_\WS G_1$ are acting on $L^2(\R^3)$.
\end{theorem}

The fact that the supercell quantities converge to the corresponding quantities of the periodic rHF model was already proved in~\cite[Theorem 4]{Cances2008}. However, following the proof of the latter article, we only find a $O\left(L^{-1}\right)$ convergence rate.

The proof of Proposition~\ref{prop:expCV} and Theorem~\ref{th:Deleurence_exp} rely on Bloch transforms.

%
%

\section{Bloch transform and supercell Bloch transform}
\label{sec:Bloch}

\subsection{Bloch transform from $L^2(\R^3)$ to $L_\per^2(\BZ, L^2(\Gamma))$}

We recall in this section the basic properties of the usual Bloch transform~\cite[Chapter XIII]{ReedSimon4}). Let $(\ba_1, \ba_2, \ba_3)$ be a basis of the lattice $\Lat$, so that $\Lat = \Z \ba_1 + \Z \ba_2 + \Z \ba_3$.
We define the dual lattice $\RLat$ by $\RLat = \Z \ba_1^\ast + \Z \ba_2^\ast + \Z \ba_3^\ast$ 
where the vectors $\ba_i^*$ are such that $\ba_i^* \cdot \ba_j = 2 \pi \delta_{ij}$. The unit cell and the reciprocal unit cell are respectively defined by
\[
	\WS := \left\{ \alpha_1 \ba_1 + \alpha_2 \ba_2 + \alpha_3 \ba_3, \ (\alpha_1, \alpha_2, \alpha_3) \in [-1/2, 1/2)^3 \right\},
\]
and
 \begin{equation*}
 	\BZ := \left\{  \alpha_1 \ba_1^* + \alpha_2 \ba_2^* + \alpha_3 \ba_3^*, \quad (\alpha_1, \alpha_2, \alpha_3) \in [-1/2, 1/2)^3 \right\}.
 \end{equation*}
 Note that $\BZ$ differs from the first Brillouin-zone when the crystal is not cubic. We consider the Hilbert space $L^2(\BZ, L^2_\per(\WS))$, endowed with the normalized inner product
 \[
 	\bra f(\bq, \bx), g(\bq, \bx) \ket_{L^2(\BZ, L^2_\per(\WS))} := \fint_{\BZ} \int_{\WS} \overline{f}(\bq, \bx) g(\bq, \bx) \, \rd \bx \, \rd \bq.
 \]
 The Bloch transform is defined by
\begin{equation} \label{eq:def:cZ}
	 \begin{array}{lcll}
	\cZ: & L^2(\R^3) & \to & L^2(\BZ, L^2_\per(\WS)) \\
		 & w & \mapsto & (\cZ w)(\bq,\bx) := w_\bq(\bx) := \displaystyle \sum\limits_{\bR \in \Lat} \re^{-\ri \bq \cdot (\bx + \bR)} w(\bx + \bR).
	\end{array}
\end{equation}
Its inverse is given by
\[
	\begin{array}{lcll}
	\cZ^{-1}: & L^2(\BZ, L^2_\per(\WS)) & \to & L^2(\R^3) \\
		& w_\bq(\bx) & \mapsto & (\cZ^{-1} w)(\bx) := \displaystyle \fint_{\BZ} \re^{ \ri \bq \cdot \bx} w_\bq( \bx) \, \rd \bq.
	\end{array}
\]
It holds that $\cZ$ is an isometry, namely
\[
	\| \cZ w \|_{L^2(\BZ, L^2_\per(\WS))}^2 = \fint_{\BZ} \int_{\WS}  \left| (\cZ w)(\bq, \bx)  \right|^2 \rd \bx \ \rd \bq =    \| w \|_{L^2(\R^3)}^2.
\]
For $\bm \in \RLat$, we introduce the unitary operator $U_\bm$ acting on $L^2_\per(\WS)$ defined by 
\begin{equation} \label{eq:def:Um}
	\forall \bm \in \RLat, \quad \forall f \in L^2_\per(\WS), \quad
	\left( U_\bm f \right)(\bx) = \re^{-\ri \bm \cdot \bx} f(\bx).
\end{equation}
From~\eqref{eq:def:cZ}, it is natural to consider $\cZ w$ as a function of $L^2_\loc \left( \R^3, L^2_\per(\WS) \right)$ such that
\begin{equation} \label{eq:function_covariant}
	 \forall w \in L^2(\R^3), \quad \forall \bm \in \RLat, \quad \forall \bq \in \BZ, \quad 
	 \left( \cZ w \right)(\bq + \bm, \cdot) = w_{\bq + \bm} = U_\bm w_{\bq} = U_\bm \left( \cZ w (\bq, \cdot) \right).
\end{equation}

Let $A$ with domain $\cD(A)$ be a possibly unbounded operator acting on $L^2_\per(\WS)$. We say that $A$ commutes with $\Lat$-translations if $\tau_\bR A = A \tau_\bR$ for all $\bR \in \Lat$. If $A$ commutes with $\Lat$-translations, then it admits a Bloch decomposition. The operator $\cZ A \cZ^{-1}$ is block diagonal, which means that there exists a family of operators $\left(A_\bq \right)_{\bq \in \BZ}$ acting on $L^2_\per(\WS)$, such that, if $f \in L^2(\R^3)$ and $g \in \cD(A)$ are such that $f = Ag$, then, for almost any $\bq \in \BZ$, $g_\bq \in L^2_\per(\WS)$ is in the domain of $A_\bq$, and
\begin{equation}\label{eq:op_diagonal}
 f_\bq = A_\bq g_\bq.
\end{equation} 
In this case, we write
\begin{equation*}
	\cZ A \cZ^{-1} = \fint_{\BZ}^{\oplus} A_\bq \rd \bq \quad \text{(Bloch decomposition of $A$).}
\end{equation*}
From~\eqref{eq:function_covariant}, we extend the definition of $A_\bq$, initially defined for $\bq \in \BZ$, to $\bq \in \R^3$, with 
\begin{equation} \label{eq:rotation}
	\forall \bm \in \RLat, \quad \forall \bq \in \BZ, \quad A_{\bq + \bm} = U_\bm A_\bq U_\bm^{-1},
\end{equation}
so that~\eqref{eq:op_diagonal} holds for almost any $\bq\in \RR^3$. 
If $A$ is locally trace-class, then $A_\bq$ is trace-class on $L^2_\per(\WS)$ for almost any $\bq \in \R^3$. The operator $A$ can be associated a density $\rho_A$, which is an $\Lat$-periodic function, given by 
\[
	\rho_A = \fint_{\BZ} \rho_{A_\bq} \rd \bq,
\] 
where $\rho_{A_\bq}$ is the density of the trace-class operator $A_\bq$. The trace per unit volume of $A$ (defined in~\eqref{eq:VTr_def1}) is also equal to
\begin{equation} \label{eq:VTr}
	\VTr(A) = \fint_{\BZ} \Tr_{L^2_\per(\WS)} \left( A_\bq \right) \rd \bq .
\end{equation}

%
%

\subsection{Bloch transform from $L^2_\per(\Gamma_L)$ to $\ell^2(\Lambda_L, L^2_\per(\Gamma))$}

We present in this section the ``supercell'' Bloch transform. This transformation goes from~$L^2_\per(\Gamma_L)$ to $\ell^2(\Lambda_L, L^2_\per(\Gamma))$, where $\Lambda_L := \left( L^{-1} \RLat\right) \cap \BZ$, \textit{i.e.} 
\begin{equation} \label{eq:def:LambdaL}
	\Lambda_L := \left\{ \dfrac{2 k_1}{L} \ba_1^* +  \dfrac{2 k_2}{L} \ba_2^* +  \dfrac{2 k_3}{L} \ba_3^*, \ (k_1, k_2, k_3) \in \left\{ \dfrac{-L+ \eta}{2},\dfrac{-L+\eta}{2} + 1, \cdots,  \dfrac{L + \eta}{2} - 1 \right\}^3 \right\},
\end{equation}
with $\eta = 1$ if $L$ is odd, and $\eta = 0$ if $L$ is even, so that there are exactly $L^3$ points in $\Lambda_L$. Similarly, we define $\Lat_L := \Lat \cap \WS_L$, 
which contains $L^3$ points of the lattice $\Lat$. The supercell Bloch transform has properties similar to those of the standard Bloch transform, the main difference being that there are only a finite number of fibers. We introduce the Hilbert space $\ell^2(\Lambda_L, L^2_\per(\WS))$ endowed with the normalized inner product
\[
	\bra f(\bQ, \bx), g(\bQ, \bx) \ket_{\ell^2(\Lambda_L, L^2_\per(\WS))} := \dfrac{1}{L^3}\sum_{\bQ \in \Lambda_L} \int_{\WS} \overline{f}(\bQ, \bx) g(\bQ, \bx) \, \rd \bx.
\]
The supercell Bloch transform is defined by
\[
	\begin{array}{llll}
	\cZ_L: & L^2_\per(\WS_L) & \to &  \ell^2(\Lambda_L, L^2_\per(\WS)) \\
		& w & \mapsto & (\cZ_L w)(\bQ,\bx) := w_\bQ(\bx) := \displaystyle \sum\limits_{\bR \in \Lat_L} \re^{-\ri \bQ \cdot (\bx + \bR)} w(\bx + \bR).
	\end{array}
\]
Its inverse is given by
\[
	\begin{array}{llll}
	\cZ_L^{-1}: & \ell^2(\Lambda_L, L^2_\per(\WS)) & \to & L^2_\per(\WS_L) \\
		& w_\bQ( \bx) & \mapsto & (\cZ_L^{-1} w)(\bx) := \displaystyle \dfrac{1}{L^3} \sum\limits_{\bQ \in \Lambda_L} \re^{ \ri \bQ \cdot \bx} w_\bQ( \bx).
	\end{array}
\]
It holds that $\cZ_L$ is an isometry, \textit{i.e.}
\[
	\| w \|_{L^2_\per(\WS_L)}^2 = \dfrac{1}{L^3} \sum_{\bQ \in \Lambda_L} \int_\WS \left| \left( \cZ_L w \right)(\bQ, \bx)\right|^2  \rd \bx.
\]
We can extend $\cZ$ to $\ell^\infty \left( L^{-1} \RLat , L^2_\per(\WS) \right)$ with
\[
	\forall w \in L^2_\per(\WS_L), \quad \forall \bm \in \RLat, \quad \forall \bQ \in \Lambda_L, \quad w_{\bQ + \bm} = U_\bm w_{\bQ},
\]
where the operator $U_\bm$ was defined in~\eqref{eq:def:Um}.

Let $A^L$ with domain $\cD \left(A^L \right)$ be an operator acting on $L^2_\per(\WS_L)$. If $A$ commutes with $\Lat$-translations, then it admits a supercell Bloch decomposition. The operator $\cZ^L A^L \cZ^L$ is block diagonal, which means that there exists a family of operators $(A_\bQ^L)_{\bQ \in \Lambda_L}$ acting on $L^2_\per(\WS)$ such that if $f = A^L g$ with $f \in L^2_\per(\Gamma_L)$ and $g \in \cD(A^L)$, then for all $\bQ \in \Lambda_L$,
\begin{equation} \label{eq:op_diagonal_sc}
 f_\bQ = A_\bQ^L g_\bQ.
\end{equation} 
We write
\[
	\cZ_L A^L \cZ_L^{-1} := \dfrac{1}{L^3} \bigoplus_{\bQ \in \Lambda_L} A^L_\bQ \quad \text{(supercell Bloch decomposition of $A^L$).}
\]
The spectrum of $A^L$ can be deduced from the spectra of $\left( A^L_\bQ \right)_{\bQ \in \Lambda_L}$ with
\begin{equation} \label{eq:spectrum_AL}
	\sigma \left( A^L \right) = \bigcup_{\bQ \in \Lambda_L} \sigma \left( A^L_\bQ \right).
\end{equation}
Similarly to~\eqref{eq:rotation}, we extend the definition of $A_\bQ$ to $L^{-1} \RLat$ with
\[
	\forall \bm \in \RLat, \quad \forall \bq \in \BZ, \quad A_{\bQ + \bm} = U_\bm A_\bQ U_{\bm}^{-1},
\]
so that~\eqref{eq:op_diagonal_sc} holds for all $\bQ \in L^{-1} \RLat$.

Finally, if the operator $A^L$ is trace-class, we define the trace per unit volume by
\begin{equation} \label{eq:VTrL}
	\VTr_L(A^L) = \dfrac{1}{L^3} \Tr_{L^2_\per(\WS_L)}(A^L) = \dfrac{1}{L^3} \sum\limits_{\bQ \in \Lambda_L} \Tr_{L^2_\per(\WS)}(A^L_\bQ),
\end{equation}
and the associated density is given by 
$
	\displaystyle \rho_{A^L} = \dfrac{1}{L^3} \sum\limits_{\bQ \in \Lambda_L} \rho_{A^L_{\bQ}}
$,
where $\rho_{A^L_\bQ}$ is the density of the trace-class operator $A^L_\bQ$.

%
%

\section{Proof of Proposition~\ref{prop:expCV}: the linear case} 
\label{subsec:proof_exp_lin}
\label{sec:proof_expCV}

The proofs of Proposition~\ref{prop:expCV} and Theorem~\ref{th:Deleurence_exp} are based on reformulating the problem using the Bloch transforms. Comparing quantities belonging to the whole space model on the one hand, and to the supercell model on the other hand amounts to comparing integrals with Riemann sums. The exponential convergence then relies on two arguments: quantities of interest are $\RLat$-periodic and have analytic continuations on a complex strip, and the Riemann sums for such functions converge exponentially fast to the corresponding integrals.

We prove in this section the exponential convergence of Proposition~\ref{prop:expCV}.

\subsection{Convergence of Riemann sums}

We recall the following classical lemma. For $A>0$, we denote by
\[
	S_A := \left\{ \bz \in \C^3, \ | \Im(\bz) |_\infty \le A \right\} = \R^3 + \ri \, [-A, A]^3.
\]
If $E$ is a Banach space and $d\in\N^\ast$, an $E$-valued function $F : \Omega \subset \C^d \to E$ is said to be (strongly) analytic if $(\nabla_\bz F)(\bz)$ exists in $E^d$ for all $\bz \in \Omega $. In the sequel, we assume without loss of generality that the vectors spanning the lattice $\RLat$ are ordered in such a way that $| \ba_1^\ast | \le | \ba_2^\ast | \le | \ba_3^\ast |$.

\begin{lemma}\label{lem:Riemann}
	Let $f: \R^3 \to \C$ be an $\RLat$-periodic function that admits an analytic continuation on $S_A$ for some $A > 0$. Then, there exists $C \in \R^+$ and $\alpha > 0$ such that
	\[
		\forall L \in \N^*, \quad \left| \fint_{\BZ} f(\bq) \, \rd \bq - \dfrac{1}{L^3} \sum_{\bQ \in \Lambda_L} f \left( \bQ \right)  \right| \le C_0  \sup\limits_{\bz \in S_A}  \left| f(\bz) \right|  \re^{- \alpha L}.
	\]
	The constants may be chosen equal to
	\begin{equation} \label{eq:choice_alpha}
	\alpha = (2/3) \pi A | \ba_3^\ast|^{-1} 
	\quad \text{and} \quad
		C_0 = 2 \left( \dfrac{3 + \re^{-2 \alpha}}{\left(1 - \re^{-\alpha} \right)^3} \right).
	\end{equation}
\end{lemma}

\begin{proof}[Proof of Lemma~\ref{lem:Riemann}]

Let $c_\bR(f) := \fint_{\BZ} f(\bq) \re^{-\ri \bR \cdot \bq} \, \rd \bq$ be the Fourier coefficients of $f$, so that
\[
	f(\bq) = \sum_{\bR \in \Lat} c_\bR(f) \re^{\ri \bq \cdot \bR}.
\]
It holds
\begin{align*}
	\left| \fint_{\BZ} f(\bq) \rd \bq - \dfrac{1}{ L^3} \sum_{\bQ \in \Lambda_L} f \left( \bQ \right)  \right| 
	& =  \left| c_\bnull(f) - \dfrac{1}{L^3} \sum_{\bQ \in \Lambda_L} \sum_{\bR \in \Lat} c_\bR(f) \re^{ \ri \bQ \cdot \bR} \right| \\
	& = \left| \sum_{\bR \in \Lat \setminus \{ \bnull \}} c_\bR(f) \left( \dfrac{1}{L^3} \sum_{\bQ \in \Lambda_L}  \re^{\ri \bQ \cdot \bR} \right)\right|.
\end{align*}
By noticing that
\[
	\sum_{\bQ \in \Lambda_L}  \re^{\ri \bQ \cdot \bR} = 
		\left\{ \begin{array}{ccc}
			0 & \text{if} & R \notin L \Lat \\
			L^3 & \text{otherwise} &,
		\end{array} \right.
\]
we obtain
\begin{equation} \label{eq:int_minus_sum}
	\left| \fint_{\BZ} f(\bq) \rd \bq - \dfrac{1}{ L^3} \sum_{\bQ \in \Lambda_L} f \left( \bQ \right)  \right| = \left| \sum_{\bR \in \Lat \setminus \{ \bnull \}} c_{L \bR}(f) \right|.
\end{equation}
If $f$ is analytic on $S_A$, we deduce from $f(\bq) = \sum_{\bR \in \Lat} c_\bR(f) \re^{\ri \bR \cdot \bq}$ that the analytic continuation of $f$ is given by
\[
	\forall \bq \in \R^3, \quad \forall \by \in [-A, A]^3, \quad f(\bq + \ri \by) = \sum_{\bR \in \Lat} c_\bR(f) \re^{\ri \bR \cdot \bq}\ \re^{- \bR \cdot \by},
\]
so that $\left\{ c_\bR(f) \re^{-\bR \cdot \by} \right\}_{\bR \in \Lat}$ are the Fourier coefficients of the $\RLat$-periodic function $\bq \mapsto f(\bq + \ri \by)$. In particular,
\begin{equation} \label{eq:exponential_decay}
	\forall \bR \in \Lat, \quad \forall \by \in [-A, A]^3, \quad \left| c_\bR(f) \right| \le \sup\limits_{\bq \in \BZ} \left| f(\bq + \ri \by) \right| \re^{\bR \cdot \by}.
\end{equation}
We make the following choice for $\by$. We write $\bR = k_1 \ba_1 + k_2 \ba_2 + k_3 \ba_3$ with $k_1, k_2, k_3 \in \Z$, and we let $1 \le m \le 3$ be the index such that $\left| k_m \right| = \left| k_j \right|_\infty$. Choosing $\by =  - \rm{sgn}(k_m) A | \ba_3^\ast |^{-1} \ba_m^\ast \in [-A, A]^3$, leads to 
\[
	\left| c_\bR(f) \right| \le \sup\limits_{\bz \in S_A} \left| f(\bz)\right| \re^{- 2 \pi A | \ba_3^\ast|^{-1} \left| k \right|_\infty} 
		\le  \sup\limits_{\bz \in S_A} \left| f(\bz)\right| \re^{ - \alpha  | k_1 |_1}
\]
where we used the inequality $| \bk |_{\infty} \ge (1/3) | \bk |_{1}$, and we set $\alpha = (2/3) \pi A | \ba_3^\ast|^{-1}$. Note that the Fourier coefficients of $f$ are exponentially decreasing. We conclude with~\eqref{eq:int_minus_sum} and the inequality
\begin{align*}
	\left| \sum_{\bR \in \Lat \setminus \{ \bnull \} } c_{L \bR }(f) \right| & \le \sum_{\bR \in \Lat  \setminus \{ \bnull \} } | c_{L \bR }(f) | 
	\le \sup\limits_{\bz \in S_A} \left| f(\bz) \right|  \sum_{\bk \in \Z^3 \setminus \{ \bnull \} } \re^{- \alpha L \left| \bk \right|_1 } 
	\\
	& \le \sup\limits_{\bz \in S_A} \left| f(\bz) \right|  \left[ \left( \sum_{k \in \Z} \re^{- \alpha L \left| k \right| } \right)^3 -1 \right] 
	=  \sup\limits_{\bz \in S_A}  \left| f(\bz) \right| \left( \dfrac{2 (3 + \re^{-2 \alpha})}{\left(1 - \re^{-\alpha} \right)^3} \right) \re^{-\alpha L} 
\end{align*}

\end{proof}


\subsection{Analyticity and basic estimates}

The exponential rates of convergence observed in~\eqref{eq:conv_energy_lin},~\eqref{eq:conv_density} and~\eqref{eq:conv_energy} will come from Lemma~\ref{lem:Riemann} for appropriate choices of functions $f$. In order to construct such functions, we notice that $H$ and $H^L$ defined in Proposition~\ref{prop:expCV} commute with $\Lat$-translations, thus admit Bloch decompositions. From
$$
\forall\bq\in \R^3,\; \left( -\Delta \right)_\bq= \left| -\ri\nabla_1 +\bq \right|^2 = \sum_{j=1}^3 \left( P_{j,1} + q_j \right)^2
\quad \text{and}\quad 
\forall \bQ \in L^{-1} \RLat,\; (-\Delta_L)_\bQ= \left(- \Delta_1 \right)_\bQ,
$$
where $\nabla_1$ denotes the gradient on the space $L^2_\per(\WS)$ and $\Delta_1$ was defined in~\eqref{eq:def:Pj}, we obtain
\begin{equation} \label{eq:Bloch_H}
	\cZ H \cZ^{-1}= \fint_{\BZ}^\oplus H_\bq \, \rd \bq
	\quad \text{with} \quad
	H_\bq := - \frac12 \left| - \ri \nabla_1 + \bq \right|^2 + V = \frac12 \left( - \Delta_1 - 2 \ri \bq \cdot \nabla_1 + | \bq |^2 \right) + V,
\end{equation}
and
\[
	\cZ_L H^L \cZ_L^{-1} = \dfrac{1}{L^3}\bigoplus_{\bQ \in \Lambda_L} H_\bQ.
\]
In other words, for all $\bQ$ in $\Lambda_L$, $\left( H^L \right)_\bQ = H_\bQ$. In addition, the spectrum of $H$ can be recovered from the spectra of $\left( H_\bq \right)_{\bq \in \BZ}$ with~\cite[Chapter XIII]{ReedSimon4}
\begin{equation*}
	\sigma(H) = \bigcup_{\bq \in \BZ} \sigma( H_\bq )
\end{equation*}
Together with~\eqref{eq:spectrum_AL} we deduce that, since $H$ has a gap of size $g$ centered around $\varepsilon_F$, then $H^L$ has a gap of size at least $g$  around~$\varepsilon_F$. \\

In the sequel, we introduce, for $\bz \in \C^3$, the operator (we denote by $\bz^2 := \sum_{j=1}^3 z_j^2$ for simplicity)
\begin{equation} \label{eq:def_Hz}
	H_\bz :=  \dfrac{1}{2} \left( - \Delta_1 - 2 \ri \bz \cdot \nabla_1 + \bz^2 \right) + V
	\quad \text{acting on} \quad
	L^2_\per(\WS).
\end{equation}
With the terminology of~\cite[Chapter VII]{Kato2012}, the map $\bz \mapsto H_\bz$ is an holomorphic family of type $(A)$.  Let $\Sigma := \inf \sigma(H)$ be the bottom of the spectrum of $H$. We consider the positively oriented simple closed loop  $\sC = \sC_1 \cup \sC_2 \cup \sC_3 \cup \sC_4$ in the complex plane, consisting of the following line segments: $\sC_1 = [ \varepsilon_F - \ri, \varepsilon_F + \ri]$, $\sC_2 = [ \varepsilon_F +\ri, \Sigma - 1 + \ri]$, $\sC_3 = [\Sigma - 1 + \ri , \Sigma - 1 - \ri]$ and $\sC_4 = [\Sigma-1-\ri, \varepsilon_F - \ri]$. 

\begin{figure}[h!]
\begin{center}
\begin{tikzpicture}
	\fill[red!60] (-3, -0.1) rectangle (0.5, 0.1);
	\fill[red!60] (1.5, -0.1) rectangle (5, 0.1);
	\draw[->] (-5, 0) -> (6, 0);
	
	\draw (-3, 0.2) -- (-3, -0.2);
	\node at (-3, -0.5) {$\Sigma$};
	
	\draw (1, -1) -- (1, 1) -- (-4, 1) -- (-4, -1) -- (1, -1);
	\draw (-1.8, 1.2) -- (-2, 1) -- (-1.8, 0.8);
	\draw (-2.2, -1.2) -- (-2, -1) -- (-2.2, -0.8);

	\node at (1.3, -0.5) {$\varepsilon_F$};
	\node at (3, 0.5) {$\sigma(H)$};
	
	\node at (1.3, 0.5) {$\sC_1$};
	\node at (-1, 1.4) {$\sC_2$};	
	\node at (-4.5, 0.5) {$\sC_3$};
	\node at (-1, -1.5) {$\sC_4$};
	
\end{tikzpicture}
\caption{The loop $\sC$.}
\label{fig:sC}
\end{center}
\end{figure}

The projectors defined in~\eqref{eq:gamma_gammaL} can be written, using the Cauchy's residue theorem, as
$$
\gamma=\dfrac{1}{2 \ri \pi} \oint_{\sC} \dfrac{\rd \lambda}{\lambda - H}\quad \text{and}\quad \gamma_L=\dfrac{1}{2 \ri \pi} \oint_{\sC} \dfrac{\rd \lambda}{\lambda - H^L}.
$$
Together with~\eqref{eq:Bloch_H}, it follows that $\gamma$ and $\gamma_L$ commutes with $\Lat$-translations, with
\begin{equation} \label{eq:gamma_lemma}
		\cZ \gamma \cZ^{-1} = \fint_{\BZ}^\oplus \gamma_\bq \rd \bq \quad \text{and} \quad 
		\cZ_L \gamma_L \cZ_L^{-1} = \dfrac{1}{L^3} \bigoplus_{\bQ \in \Lambda_L} \gamma_\bQ,
	\end{equation}
	where we set
	\begin{equation} \label{eq:gammaq_lemma}
		\forall \bq \in \R^3, \quad \gamma_\bq := \dfrac{1}{2 \ri \pi} \oint_{\sC} \dfrac{\rd \lambda}{\lambda - H_\bq}.
\end{equation}
For $\bQ \in \RLat$, it holds $(\gamma_L)_\bQ = \gamma_\bQ$. The analytic continuation of~\eqref{eq:gammaq_lemma} is formally
\[
	\forall \bz \in \C^3, \quad \gamma_{\bz} := \dfrac{1}{2 \ri \pi} \oint_{\sC} \dfrac{\rd \lambda}{\lambda - H_\bz}.
\]
The fact that $\lambda - H_\bz$ is indeed invertible, at least for $\bz$ in some $S_A$ for $A > 0$ is proved in the following lemma. For $\bz \in \C^3$, and $\lambda \in \sC$, we introduce
\begin{equation} \label{eq:def:B1_B2}
	{B_1}(\lambda, \bz) := (1 - \Delta_1) \dfrac{1}{\lambda - H_\bz}
	\quad \text{and} \quad
	{B_2}(\lambda, \bz) := \dfrac{1}{\lambda - H_\bz} (1 - \Delta_1).
\end{equation}

\begin{lemma} \label{lem:bounds_B}
	For all $\bq \in \R^3$, and all $\lambda \in \sC$, the operator $\lambda - H_\bq$ is invertible, and there exists a constant $ {C_1} \in \R^+$ such that,
	\begin{equation} \label{eq:bounds_B1_B2}
		\forall \bq \in \BZ, \quad \forall \lambda \in \sC, \quad 
	\left\|  {B_1}(\lambda, \bq) \right\|_{\cB(L^2_\per(\WS))} \le {C_1}
		\quad \text{and} \quad
	\left\|  {B_2}(\lambda, \bq) \right\|_{\cB(L^2_\per(\WS))} \le {C_1}.
	\end{equation}
	Denoting by $| \BZ|_2 := \sup \left\{ | \bq |_2, \ \bq \in \BZ \right\}$, we can choose
	\begin{equation} \label{eq:def_C}
		{C_1} = 4 + \dfrac{2 + 4 | \BZ |_2^2 + 8 \| V \|_{L^\infty} + 8 \varepsilon_F}{\min(1,g)} .
	\end{equation}
	Moreover, there exists $A > 0$ such that, for all $\bz \in S_A$ and all $\lambda \in \sC$, the operator $\lambda - H_\bz$ is invertible, and there exists a constant $ C_2 \in \R^+$ such that
	\begin{equation} \label{eq:bounds_B_z}
		\forall \bz \in \BZ + \ri \, [-A, A]^3, \quad \forall \lambda \in \sC, \quad 
		\left\|  {B_1}(\lambda, \bz) \right\|_{\cB(L^2_\per(\WS))}  \le {C_2}
		\quad \text{and} \quad
		\left\|   {B_2}(\lambda, \bz) \right\|_{\cB(L^2_\per(\WS))}\le {C_2}.
	\end{equation}
	We can choose
	\begin{equation} \label{eq:choice_A}
		A = \min \left( 1, \ \dfrac{1}{2 C_1 (1 + | \BZ |_2)} \right)
		\quad \text{and} \quad
		{C_2} = 2 C_1.
	\end{equation}
\end{lemma}

This lemma was proved in the one-dimensional case by Kohn in~\cite{Kohn1959}, and similar results were discussed by Des Cloizeaux in~\cite{Cloizeaux1964, Cloizeaux1964a}.

\begin{remark}
The bounds~\eqref{eq:bounds_B1_B2} and~\eqref{eq:bounds_B_z} are not uniform for $\bq \in \R^3$ (only for $\bq \in \BZ$). This comes from the fact that, for $\bm \in \RLat$,
\[
	\left\| \dfrac{1 - \Delta_1}{1 + (- \ri \nabla_1 + \bm)^2} \right\|_{\cB(L^2_\per(\WS))} 
		\ge \left\bra \dfrac{\re^{- \ri \bm \cdot \bx}}{| \WS |^{1/2}} \left| \dfrac{1 - \Delta_1}{1 + (- \ri \nabla_1 + \bm)^2} \right| \dfrac{\re^{- \ri \bm \cdot \bx}}{| \WS |^{1/2}} \right\ket = \dfrac{1 + \left| \bm \right|^2}{| \WS |}.
\]
\end{remark}

\begin{remark}
From Lemma~\ref{lem:bounds_B}, we deduce that $\left( \gamma_\bz \right)_{\bz \in S_A}$ is an analytic family of bounded operators. Since $\gamma_\bq$ is an orthogonal projector for $\bq \in \R^3$, \textit{i.e.} $\gamma_\bq = \gamma_\bq \gamma_\bq$, we deduce that $\gamma_\bz = \gamma_\bz \gamma_\bz$ for all $\bz \in S_A$, so that $\gamma_\bz$ is a (not necessarily orthogonal) projector. Also, $\Tr(\gamma_\bz)$ is a constant independent of $\bz \in S_A$.
\end{remark}

\begin{proof}[Proof of Lemma~\ref{lem:bounds_B}]
	From the inequality $|a |^2 \le 2 | a +b|^2 + 2 | b |^2$, we get that, for $\bq \in \BZ$, it holds $\left| - \ri \nabla_1 + \bq \right|^2 + | \bq |^2 \ge - \frac12 \Delta_1 $. We deduce that
	\begin{equation} \label{eq:bound_Hq}
		\forall \bq \in \BZ, \quad H_\bq \ge - \frac14 \Delta_1 - \frac12  | \BZ |_2^2 - \| V \|_{L^\infty}.
	\end{equation}
	We first consider the part $\sC_1$ of the contour $\sC$ (see Figure~\ref{fig:sC}). It holds
	\begin{equation} \label{eq:bound_Hq_bis}
		\forall \lambda \in \sC_1, \quad \forall \bq \in \BZ, \quad
		| H_\bq - \lambda |^2 \ge \left| \Re \left( H_\bq - \lambda \right) \right|^2 = \left| H_\bq - \varepsilon_F \right|^2.
	\end{equation}
	Since $\left| H_\bq - \varepsilon_F \right| \ge g/2$, we get
	\begin{equation} \label{eq:Hq_C1_1}
		\forall \lambda \in \sC_1, \quad \forall \bq \in \BZ, \quad
		| H_\bq - \lambda | \ge g/2.
	\end{equation}
	On the other hand, from~\eqref{eq:bound_Hq} and~\eqref{eq:bound_Hq_bis}, it holds that
	\begin{equation} \label{eq:Hq_C1_2}
		\forall \lambda \in \sC^1, \quad \forall \bq \in \BZ, \quad
		| H_\bq - \lambda | \ge H_\bq - \varepsilon_F \ge  - \frac14 \Delta_1 - \frac12  | \BZ |_2^2 - \| V \|_{L^\infty} - \varepsilon_F.
	\end{equation}
	Combining~\eqref{eq:Hq_C1_1} and~\eqref{eq:Hq_C1_2} leads to
	\[
		\forall M \ge 0, \quad \forall \lambda \in \sC_1, \quad \forall \bq \in \BZ, \quad
		(M + 4) | H_\bq - \lambda | \ge - \Delta_1 + M \frac{g}{2} - 2| \BZ |_2^2 - 4 \| V \|_{L^\infty} - 4 \varepsilon_F.
	\]
	Choosing $M = (2 + 4 | \BZ |_2^2 + 8 \| V \|_{L^\infty} + 8 \varepsilon_F)/g$ gives
	\[
		\forall \lambda \in \sC^1, \quad \bq \in \BZ,
		\quad |H_\bq - \lambda | \ge \left( 4 + \dfrac{2 + 4 | \BZ |_\infty^2 + 8 \| V \|_{L^\infty} + 8 \varepsilon_F}{g}\right)^{-1} (1 - \Delta_1),
	\]
	which proves~\eqref{eq:bounds_B1_B2} for $\lambda \in \sC_1$. The inequalities on the other parts of $\sC$ are proved similarly, the inequalities~\eqref{eq:Hq_C1_1} and~\eqref{eq:Hq_C1_2} being respectively replaced by their equivalent
	\begin{align*}
		\forall \lambda \in \sC_2 \cup \sC_4, \quad 
		| H_\bq - \lambda |^2 \ge \left| \Im \left( H_\bq - \lambda \right) \right|^2 \ge 1 \quad \text{and} \quad | H_\bq - \lambda | \ge H_\bq - \Sigma - 1 \ge H_\bq - \varepsilon_F, \\
		\forall \lambda \in \sC_3, \quad 
		| H_\bq - \lambda |^2 \ge \left| \Re \left( H_\bq - \lambda \right) \right|^2 \ge 1 \quad \text{and} \quad | H_\bq - \lambda | \ge H_\bq - \Sigma - 1 \ge H_\bq - \varepsilon_F.
	\end{align*}
	This proves~\eqref{eq:bounds_B1_B2}. We now prove~\eqref{eq:bounds_B_z}. For $\bz=\bq+\ri \by \in \C^3$ with $\bq \in \BZ$ and $\by \in \R^3$, one can rewrite~\eqref{eq:def_Hz} as
 \[
	 H_\bz = H_\bq + \by \cdot \nabla_1 + \ri \bq \cdot \by - \frac12| \by|^2 = H_\bq + \by \cdot \left( \nabla_1 - \frac12 \by + \ri \bq \right).
 \]
In particular,
\begin{align}
 \lambda-H_\bz&= \lambda -H_\bq+H_\bq-H_\bz= ( \lambda - H_\bq) \left(1 - (\lambda - H_\bq)^{-1} \left[ \by \cdot \left(\nabla_1 - \frac12 \by + \ri \bq \right) \right] \right) \nonumber \\
 &= (\lambda - H_\bq) \left( 1 - B_2(\lambda, \bq) \dfrac{1}{1 - \Delta_1}\left[ \by \cdot \left( \nabla_1 - \frac12 \by + \ri \bq \right) \right] \right). \label{eq:lambda_Hz}
\end{align}
For $| \by |_\infty \le 1$, we have
\[
	\left| \dfrac{1}{1 - \Delta_1}\left[ \by \cdot \left( \nabla_1 - \frac12 \by + \ri \bq \right) \right]  \right| 
	\le \left| \by \right|_{\infty} \left( \left| \dfrac{| \nabla_1 |}{1 - \Delta_1} \right| + \frac12 | \by |_\infty + | \bq |_\infty \right) 
	\le \left| \by \right|_{\infty} \left( 1 +| \BZ|_2 \right).
\]
Together with~\eqref{eq:bounds_B1_B2}, we obtain that for all $| \by |_\infty \le A := \min \left( 1, (2 C_1 (1 + | \BZ |_2))^{-1} \right)$, 
\[
	\left\| B_2(\lambda, \bq) \dfrac{1}{1 - \Delta_1}\left[ \by \cdot \left( \nabla_1 - \frac12 \by + \ri \bq \right) \right] \right\| \le \frac12.	
\]
As a result, from~\eqref{eq:lambda_Hz}, we get that for all $\bq \in \BZ$ and all $\by \in [-A, A]$, the operator $\lambda - H_\bz$ is invertible, with
\[
	\left\| \dfrac{1}{\lambda - H_\bz} (1 - \Delta_1) \right\|_{\cB(\cH_1)} \le C_2 := 2C_1
	\quad \text{and} \quad
	\left\| (1 - \Delta_1) \dfrac{1}{\lambda - H_\bz} \right\|_{\cB(\cH_1)} \le C_2.
\]
\end{proof}

For $\bz \in S_A$, we introduce the operators $\widetilde{B_1}(\bz)$ and $\widetilde{B_2}(\bz)$ respectively defined by
\begin{equation} \label{eq:def:widetilde_B}
		\widetilde{B_1}(\bz) := (1 - \Delta_1) \gamma_\bz 
		\quad \text{and} \quad
		\widetilde{B_2}(\bz) := \gamma_\bz (1 - \Delta_1).
\end{equation}
In the sequel, for $k \in \N^*$, we denote by $\fS_k(\cH)$ the $k$-th Schatten class~\cite{Simon2005} of the Hilbert space~$\cH$~;~$\fS_1(\cH)$ is the set of trace-class operators, and $\fS_2(\cH)$ is the set of Hilbert-Schmidt operators. From Lemma~\ref{lem:bounds_B}, we obtain the following result.

\begin{lemma} \label{lem:widetilde_B}
	There exists a constant $C_3 \in \R^+$ such that
	\[
		\forall \bz \in \BZ +  \ri [-A, A]^3, \quad
		\left\| \widetilde{B_1} (\bz) \right\|_{\cB(L^2_\per(\WS))} \le C_3
		\quad \text{and} \quad
		\left\| \widetilde{B_2}(\bz) \right\|_{\cB(L^2_\per(\WS))} \le C_3.
	\]
	The value of $C_3$ can be chosen equal to
	\[
		C_3 = \dfrac{1}{\pi} C_1 \left( 3 + \varepsilon_F + \| V \|_{L^\infty} \right).
	\]
	Also, for all $\bz \in \BZ + \ri [-A, A]^3$, the operator $\gamma_\bz$ is trace-class, and
	\begin{equation} \label{eq:def:C4}
		\left\| \gamma_\bz \right\|_{\fS_1( L^2_\per(\WS))} \le C_4
		\quad \text{with} \quad
		C_4 = C_1^2 \sum_{\bk \in \RLat} \left( \dfrac{1}{1 + | \bk |^2} \right)^2.
	\end{equation}
\end{lemma}

\begin{proof}
	The first assertion comes from the fact that
	\[
		 \widetilde{B_1}(\bz) = \dfrac{1}{2 \ri \pi} \oint_{\sC} B_1(\lambda, \bz) \ \rd \lambda,
	\]
	and the fact that $\left| \sC \right|   = 6 + 2 (\varepsilon_F - \Sigma)$ (see Figure~\ref{fig:sC}). Note that since $ \left| - \ri \nabla_1 + \bq  \right|^2 \ge 0$, it holds $\Sigma \ge - \| V \|_{L^\infty}$. To get the second assertion, we note that $\gamma_\bz$ is a projector, so that
	\[
		\gamma_\bz = \gamma_\bz \gamma_\bz = \widetilde{B_2}(\bz) \left( \dfrac{1}{1 - \Delta_1} \right)^2 \widetilde{B_1}(\bz).
	\]
	The operator $(1 - \Delta_1)^{-2}$ being trace-class, with
	\[
		\left\| (1 - \Delta)^{-2} \right\|_{\fS_1(L^2_\per(\WS))} = \sum_{\bk \in \RLat} \left( \dfrac{1}{1 + | \bk |^2} \right)^2,
	\]
	we obtain~\eqref{eq:def:C4}.
\end{proof}


\subsection{Convergence of the kinetic energy per unit volume}

The kinetic energy per unit volume of the states $\gamma$ and $\gamma_L$ defined in~\eqref{eq:gamma_gammaL} are respectively given by
\[
	K_\per := \VTr \left( - \Delta \gamma \right) 
	\quad \text{and} \quad
	K_L := \VTr_L \left( - \Delta_L \gamma_L \right).
\]
Using the Bloch decomposition of $\gamma$ and $\gamma_L$ in~\eqref{eq:gamma_lemma}-\eqref{eq:gammaq_lemma}, and the properties~\eqref{eq:VTr} and~\eqref{eq:VTrL}, we obtain that
\[
	K^\per = \sum_{j=1}^3 \fint_{\BZ} K_j(\bq) \, \rd \bq
	\quad \text{and} \quad
	K^L = \sum_{j=1}^3 \dfrac{1}{L^3}\sum_{\bq \in \Lambda_L}  K_j(\bQ)
\]
where, for $1 \le j \le 3$, we introduced the function
\[
	K_j : \bq \ni \R^3 \mapsto \Tr_{L^2_\per(\WS)} \left( \left(P_{j} + q_j \right) \gamma_\bq \left(P_{j} + q_j \right) \right).
\]
Here, we denoted by $P_j := P_{1,j}$ for simplicity. Recall that the operator $P_{1,j}$ was defined in~\eqref{eq:def:Pj}.
The error on the kinetic energy per unit volume $K^\per - K^L$ is therefore equal to the difference between integrals and corresponding Riemann sums. In the sequel, we introduce, for $1 \le j \le 3$, the function
\[
	\forall \bz \in S_A, \quad K_j(\bz) := \Tr_{L^2_\per(\WS)} \left( \left(P_{j} + z_j \right) \gamma_\bz \left(P_{j} + z_j \right) \right).
\]

\begin{lemma}[Exponential convergence of the kinetic energy] \label{lem:conv_kinetic}
	For all $1 \le j \le 3$, the function $K_j$ is $\Lat$-periodic, and admits an analytic continuation on $S_A$, where $A > 0$ was defined in~\eqref{eq:choice_A}. Moreover, it holds
	\begin{equation} \label{eq:bound_Kj}
		\sup_{\bz \in S_A} \left| K_j(\bz) \right| \le C_5 \quad \text{where} \quad C_5 = \left( | \BZ |_2 + A + \frac12 \right)^2 C_3^2 C_4.
	\end{equation}
	As a consequence, from Lemma~\ref{lem:Riemann}, it holds 
	\[
		\left| K^\per - K^L \right| \le C_0 C_5 \re^{-\alpha L},
	\]
	where $C_0 \in \R^+$ and $\alpha > 0$ were defined in~\eqref{eq:choice_alpha}.
\end{lemma}

\begin{proof}

The $\Lat$-periodicity comes from the covariant identity~\eqref{eq:rotation}. To prove the analyticity, it is enough to prove that $\partial_{z_k} \left( (P_{j} + z_j) \gamma_{\bz} (P_{j} + z_j) \right)$ is a trace-class operator for all $\bz \in \BZ + \ri [-A, A]^3$. We only consider the case $j = 1$ and $k = 1$, the other cases being similar.  We have
\begin{equation} \label{eq:nabla_z}
 	\partial_{z_1} \left( (P_1 + z_1) \gamma_{\bz} (P_1 + z_1) \right)  = \gamma_{\bz} (P_1 + z_1) + (P_1 + z_1) (\partial_{z_1} \gamma_{\bz}) (P_1 + z_1)+ (P_1 + z_1) \gamma_{\bz}.
\end{equation}
We first show that $(P_1 + z_1) \gamma_{\bz}$ is a bounded operator. We have
\begin{align*}
 (P_1 + z_1) \gamma_{\bz}&=  \dfrac{ (P_1 + z_1)}{1 - \Delta_1} \widetilde{B_1}(\bz),
\end{align*}
where $\widetilde{B_1}$ was defined in~\eqref{eq:def:widetilde_B}. From Lemma~\eqref{lem:widetilde_B} and the fact that $(P_1 + z_1)(1-\Delta_1)^{-1}$ is a bounded operator, we deduce that $(P_1 + z_1) \gamma_{\bz}$ is bounded. The proof is similar for the operator $\gamma_\bz (P_1 + z_1)$.
We now turn to the middle term of~\eqref{eq:nabla_z}. Since $\gamma_\bz$ is a projector, it holds $\gamma_\bz = \gamma_\bz \gamma_\bz$. We obtain
\begin{align*}
 	& (P_1 + z_1) (\partial_{z_1} \gamma_{\bz}) (P_1 + z_1)  = (P_1 + z_1) \left[  \gamma_\bz(\partial_{z_1} \gamma_{\bz})+( \partial_{z_1} \gamma_{\bz}) \gamma_\bz\right](P_1 + z_1) \\
 & \qquad = \left[ (P_1 + z_1)\gamma_\bz\right] \gamma_\bz (\partial_{z_1}\gamma_{\bz})(P_1 + z_1)+ 
 (P_1 + z_1)(\partial_{z_1}\gamma_{\bz}) \gamma_\bz  \left[ \gamma_\bz(P_1 + z_1)\right].
\end{align*}
We already proved that the operators $(P_1 + z_1)\gamma_\bz$ and $\gamma_\bz(P_1 + z_1)$ were bounded. Also, $\gamma_\bz$ is a trace-class operator. To prove that $(P_1 + z_1) \left( \partial_{z_1} \right) \gamma_{\bz} (P_1 + z_1)$ is trace class, it is therefore sufficient to show that $ (P_1 + z_1)( \partial_{z_1} \gamma_{\bz})$ is bounded. We have
\begin{align*}
  (P_1 + z_1)(\partial_{z_1} \gamma_{\bz})&= \dfrac{1}{2 \ri \pi} \oint_\sC (P_1 + z_1)\frac{1}{\lambda-H_\bz} (P_1 + z_1)\frac{1}{\lambda-H_\bz}\rd \lambda\\
  &= \dfrac{1}{2 \ri \pi} \oint_\sC\left(\frac{P_1 + z_1}{1-\Delta_1}B_1(\lambda, \bz) \right)^2 \rd\lambda,
 \end{align*}
which is a bounded operator. We conclude that $\partial_{z_1} \left( (P_1 + z_1) \gamma_{\bz} (P_1 + z_1) \right)$ is a trace-class operator. Finally, for $1 \le j \le 3$, $K_j$ is an analytic function on $S_A$. \\

To get the bound~\eqref{eq:bound_Kj}, we write that
\begin{align*}
	K_j(\bz) & = \Tr_{L^2_\per(\WS)} \left( (P_j + z_j) \gamma_\bz (P_j + z_j)  \right) 
		= \Tr_{L^2_\per(\WS)} \left( (P_j + z_j) \gamma_\bz \gamma_\bz \gamma_\bz (P_j + z_j)  \right) \\
		& =  \Tr_{L^2_\per(\WS)} \left( \dfrac{P_j + z_j}{1 - \Delta_1} \widetilde{B_1}(\bz) \gamma_\bz \widetilde{B_2}(\bz)  \dfrac{P_j + z_j}{1 - \Delta_1}  \right).
\end{align*}
The bound~\eqref{eq:bound_Kj} easily follows from Lemma~\ref{lem:widetilde_B} and the estimate
\[
	\forall \bz \in \BZ + \ri [-A, A]^3, \quad \left\| \dfrac{P_j + z_j}{1 - \Delta_1} \right\|_{\cB(L^2_\per(\WS))} 
	\le | z_j | + \left\| \dfrac{P_j }{1 - \Delta_1} \right\|_{\cB(L^2_\per(\WS))}  \le \left| \BZ \right|_2 + A + \frac12.
\]

\end{proof}



\subsection{Convergence of the ground state density}

We now prove~\eqref{eq:conv_density}. The densities of $\gamma$ and $\gamma_L$ defined in~\eqref{eq:gamma_lemma}-\eqref{eq:gammaq_lemma} are respectively
\[
	\rho_{\gamma} := \fint_{\BZ} \rho_{\gamma_{\bq}} \ \rd \bq 
	\quad \text{and} \quad
	\rho_{\gamma_L} := \dfrac{1}{L^3} \sum_{\bQ \in \Lambda_L} \rho_{\gamma_{\bQ}}.
\]
In particular, if $W$ is a regular $\Lat$-periodic trial function, it holds that
\[
	M_W^\per := \int_{\WS} \rho_\gamma W  = \fint_{\BZ} \Tr_{L^2_\per} \left( \gamma_\bq W \right) \rd \bq 
	\quad \text{and} \quad
	M_W^L := \int_{\WS} \rho_{\gamma_L} W = \dfrac{1}{L^3} \sum_{\bQ \in \Lambda_L} \Tr_{L^2_\per} \left( \gamma_\bQ W  \right),
\]
so that the error $M_W^\per - M_W^L$ is again the difference between an integral and a corresponding Riemann sum. We introduce, for $W \in L^1_\per(\WS)$ the function
\begin{equation} \label{eq:def:MW}
	\forall \bz \in S_A, \quad M_W(\bz) := \Tr_{L^2_\per} \left( \gamma_\bz W \right).
\end{equation}

\begin{lemma} \label{lem:conv_density}
	For all $W \in L^1_\per(\WS)$, $M_W$ defined in~\eqref{eq:def:MW} is well-defined $\RLat$-periodic analytic function on $S_A$, where $A > 0$ was defined in~\eqref{eq:choice_A}, and it holds that
	\begin{equation} \label{eq:C6}
		\sup_{\bz \in S_A} \left| M_W(\bz) \right| \le C_6 \| W \|_{L^1_\per(\WS)}
		\quad \text{with} \quad
		C_6 = C_3^2   \sum_{\bk \in \RLat} \left( \dfrac{1}{1 + | \bk |^2} \right)^2 .
	\end{equation}
\end{lemma}
As a consequence, from Lemma~\ref{lem:Riemann}, it holds that
\begin{equation} \label{eq:conv_density_after_lemma}
		\left\| \rho_\gamma - \rho_{\gamma_L} \right\|_{L^\infty(\WS)} \le C_0 C_6 \re^{-\alpha L},
\end{equation}
where $C_0$ and $\alpha$ were defined in~\eqref{eq:choice_alpha}.

\begin{proof}[Proof of lemma~\ref{lem:conv_density}]

We first prove that $M_W$ is well defined whenever $W \in L^1_\per(\WS)$. For $W\in L^1_\per(\Gamma)$, we have
\begin{align*}
M_W(\bz) & =   \Tr_{L^2_\per(\WS)} \left( \gamma_\bz W \right) 
	= \Tr_{L^2_\per(\WS)} \left( \gamma_\bz W \gamma_\bz \right) = 
	\Tr_{L^2_\per(\WS)} \left( \widetilde{B_2}(\bz) (1 - \Delta_1)^{-1} W (1 - \Delta_1)^{-1} \widetilde{B_1}(\bz) \right).
\end{align*}
According to the Kato-Seiler-Simon inequality~\cite[Theorem 4.1]{Simon2005}\footnote{The proof in~\cite{Simon2005} is actually stated for operators acting on $L^p(\R^3)$. However, the proof applies straightforwardly to our bounded domain case $L^p_\per(\WS)$.}, the operator
 $(1 - \Delta_1)^{-1}\sqrt{ |W|} $ is Hilbert-Schmidt (\textit{i.e.} in the Schatten space $\fS_2(L^2_\per(\WS))$), and satisfies
\[
	\left\|  (1 - \Delta_1)^{-1}\sqrt{ |W|} \right\|_{\fS_2(L^2_\per(\WS))} \le \left(\sum_{\bk \in \RLat} \left( \dfrac{1}{1 + | \bk |^2} \right)^2 \right)^{1/2} \| W \|_{L^{1}(\WS)}^{1/2}.
\]
It follows that $ (1 - \Delta_1)^{-1} W  (1 - \Delta_1)^{-1}$ is in $\fS_1(L^2_\per(\WS))$ with 
\begin{equation} \label{eq:KSS_W_L1}
	\left\|  (1 - \Delta_1)^{-1} W  (1 - \Delta_1)^{-1} \right\|_{\fS_1(L^2_\per(\WS))} \le  \left(\sum_{\bk \in \RLat} \left( \dfrac{1}{1 + | \bk |^2} \right)^2 \right) \| W \|_{L^{1}(\WS)},
\end{equation}
The proof of~\eqref{eq:C6} then follows from Lemma~\ref{lem:widetilde_B}.

Let us now prove that, for $W \in L^1(\WS)$, $M_W$ is analytic on $S_A$. To do so, it is sufficient to show that, for $1 \le k \le 3$,  $\partial_{z_k} (\gamma_\bz W\gamma_\bz)$ is a trace class operator. We do the proof for $k = 1$. We have
\begin{align*}
\partial_{z_1} (\gamma_\bz W\gamma_\bz)&= (\partial_{z_1}\gamma_\bz)W \gamma_\bz+\gamma_\bz W (\partial_{z_1} \gamma_\bz) \\
 & = \dfrac{1}{(2 \ri \pi)^2} \oint_{\sC} \oint_{\sC}B_2(\lambda, \bz) \dfrac{1}{1 - \Delta_1} (P_1 + z_1) B_2(\lambda, \bz) \dfrac{1}{1 - \Delta_1} W \dfrac{1}{1 - \Delta_1} B_1(\lambda', \bz) \rd \lambda \rd \lambda'\\
  &+  \dfrac{1}{(2 \ri \pi)^2} \oint_{\sC} \oint_{\sC}  B_2(\lambda, \bz) \dfrac{1}{1 - \Delta_1} W \dfrac{1}{1 - \Delta_1} B_1(\lambda', \bz)(P_1 + z_1)\dfrac{1}{1 - \Delta_1} B_1(\lambda', \bz) \rd \lambda \rd \lambda'.
\end{align*}
We deduce as in the proof of Lemma~\ref{lem:conv_kinetic} that $ \nabla_\bz(\gamma_\bz W\gamma_\bz)$ is trace class, which concludes the proof.
\end{proof}

\subsection{Proof of Proposition~\ref{prop:expCV} and Corollary~\ref{cor:expCV}}
\label{sssec:proof_cor}

We now proceed with the proof of Proposition~\ref{prop:expCV}. The assertion~\eqref{eq:conv_density} was proved in Lemma~\ref{lem:conv_density}. To get~\eqref{eq:conv_energy_lin}, we write that
\[
	\VTr \left( \gamma H \right) = \frac12 \VTr \left( - \Delta \gamma \right) + \VTr \left( V \gamma \right)
	\quad \text{and} \quad
	\VTr_L \left( \gamma_L H^L \right) = \frac12 \VTr_L \left( - \Delta_L \gamma \right) + \VTr_L \left( V \gamma_L \right),
\]
so that
\[
	\left| \VTr \left( \gamma H \right) - \VTr_L \left( \gamma_L H^L \right) \right| \le \frac12 \left| \VTr \left( - \Delta \gamma \right) - \VTr_L \left( - \Delta_L \gamma \right) \right| + \| V \|_{L^1_\per(\WS)} \left\| \rho_\gamma - \rho_{\gamma_L} \right\|_{L^\infty_\per(\WS)}.
\]
The proof of~\eqref{eq:conv_energy_lin} then follows from Lemma~\ref{lem:conv_kinetic} and~\eqref{eq:conv_density_after_lemma}. \\

We now prove Corollary~\ref{cor:expCV}. We compare the total energies
\begin{align} \label{eq:compare_te}
	& \cE_\per^{\mu_\per} (\gamma) - \cE_L^{\mu_\per}(\gamma_L) =  \dfrac12 \left( \VTr (- \Delta \gamma) - \VTr_L (- \Delta_L \gamma_L) \right) \nonumber \\
	& \quad  + \dfrac12 \left( D_1 (\rho_\gamma - \mu_\per, \rho - \mu_\per) - D_1 (\rho_{\gamma_L}-\mu_\per, \rho_{\gamma_L}-\mu_\per) \right),
\end{align}
and notice that
\begin{align} \label{eq:notice}
	& \left| D_1 (\rho_\gamma - \mu_\per, \rho_{\gamma} - \mu_\per) - D_1 (\rho_{\gamma_L}-\mu_\per, \rho_{\gamma_L}-\mu_\per) \right|
	=  \left| D_1 (\rho_\gamma - \rho_{\gamma_L}, \rho_\gamma + \rho_{\gamma_L} - 2\mu_\per) \right| 
\end{align}
Using for instance the inequality (recall that $\left| \ba_1^\ast \right| \le \left| \ba_2^\ast \right| \le \left| \ba_3^\ast \right|$)
 \begin{align} \label{ineq:HLS}
 	\forall f,g \in L^2_\per(\WS), \quad
	\left| D_1(f,g) \right|  & =  \left| \sum_{\bR \in \Lat \setminus \{ \bnull \}} \dfrac{ \overline{c_\bk\left( f \right)} c_\bk \left( g \right)}{\left| \bk \right|^2} \right| 
	 \le \dfrac{1}{\left| \ba_3^\ast \right|^2}\sum_{\bR \in \Lat \setminus \{ \bnull \}} \left| \overline{c_\bk\left( f \right)} c_\bk \left( g \right) \right| \nonumber \\
	& \le \dfrac{1}{\left| \ba_3^\ast  \right|^2 \left| \WS \right|} \| f \|_{L^2_\per(\WS)} \| g \|_{L^2_\per(\WS)},
 \end{align}
 and combining~\eqref{eq:compare_te},~\eqref{eq:notice} and~\eqref{ineq:HLS}, we obtain
 \begin{align*}
 	\left| \cE_\per^{\mu_\per} (\gamma) - \cE_L^{\mu_\per}(\gamma_L) \right| & \le
		\dfrac12 \left| \VTr (- \Delta \gamma) - \VTr_L (- \Delta_L \gamma_L) \right| \\
			& \quad + \dfrac{1}{2 | \ba_3^* |^2 | \WS |} \left\| \rho_\gamma - \rho_{\gamma_L} \right\|_{L^2_\per(\WS)} \left\| \rho_\gamma + \rho_{\gamma_L} - 2 \mu_\per \right\|_{L^2_\per(\WS)}.
 \end{align*}
Corollary~\ref{cor:expCV} is therefore a consequence of Lemma~\ref{lem:conv_kinetic},~\eqref{eq:conv_density_after_lemma} and the embedding $L^\infty_\per(\WS) \hookrightarrow L^{2}_\per(\WS)$.

 
 \section{Proof for the nonlinear reduced Hartree-Fock case}
 \label{subsec:proof_exp}
 
In this section, we prove the exponential rate of convergence of the supercell model to the periodic model in the nonlinear rHF case (see Theorem~\ref{th:Deleurence_exp}). The proof consists of three steps. 

\subsection*{Step 1: \textit{Convergence of the ground-state energy per unit volume} }

In the sequel, we denote by $V_0 := \left( \rho_{\gamma_0} - \mu_\per \right) \ast_\WS G_1$ and $V_{L,0} := \left( \rho_{\gamma_{L,0}} - \mu_\per \right) \ast_\WS G_1$ (see also~\eqref{eq:sc_supercell} and~\eqref{eq:sc}).
We recall that
 \begin{align*}
	& H_0 = -\dfrac12 \Delta + V_0 \quad \text{and} \quad \gamma_0 = \mathds{1} (H_0 < \varepsilon_F) \quad \text{act on} \quad L^2(\R^3), \\
	& H_{L,0}= -\dfrac12 \Delta_L + V_{L,0}  \quad \text{and} \quad \gamma_{L,0} =  \mathds{1} (H_{L,0} < \varepsilon_F) \quad \text{act on} \quad L^2_\per(\Gamma_L).
\end{align*}
We denote by $g > 0$ the gap of $H_0$ around the Fermi level $\varepsilon_F$. \\

It was proved in~\cite{Cances2008} that the sequence $\left( V_{L,0} \right)_{L \in \N^*}$ converges to $V_{0}$ in $L^\infty_\per(\WS)$. We will prove later that this convergence is actually exponentially fast. As a result, we deduce that for $L$ large enough, say $L \ge L^\gap$, the operator $H_{L,0}$ is gapped around $\varepsilon_F$, and one may choose the Fermi level of the supercell $\varepsilon_F^L$ defined in~\eqref{eq:sc_supercell} equal to $\varepsilon_F$. We denote by $g_L$ the size of the gap of $H_{L,0}$ around $\varepsilon_F$. Without loss of generality we may assume that $L^\gap$ is large enough so that
\[
	\forall L \ge L^\gap, \quad g_L \ge \frac{g}{2}.
\]

In the last section, we proved that the constants $C \in \R^+$ and $\alpha > 0$ appearing in Proposition~\ref{prop:expCV} are functions of the parameters $\Lat$, $ \|V \|_{L^\infty}$, $g$ and $\varepsilon_F$ of the problem only. In particular, it is possible to choose $C \in \R^+$ and $\alpha > 0$ such that, for any choice of potentials $V$ among $\left\{ V_0, \left( V_{L,0} \right)_{L \ge L^\gap} \right\}$, the inequalities~\eqref{eq:conv_energy_lin}, \eqref{eq:conv_density} and~\eqref{eq:conv_energy} hold true.\\

We first consider $V= V_0$ in Proposition~\ref{prop:expCV}. We denote by $\gamma_L \in \cP_L$ the one-body density matrix defined in~\eqref{eq:gamma_lemma} for this choice of potential. Together with Corollary~\ref{cor:expCV}, we get
\[
	\forall L \in \N^*, \quad L^{-3} I_L^{\mu_\per} = L^{-3} \cE_L^{\mu_\per}(\gamma_{L,0}) \le L^{-3} \cE_L^{\mu_\per}(\gamma_L) \le \cE^{\mu_\per}_\per(\gamma_0) + C \re^{- \alpha L} = I_\per^{\mu_\per} + C \re^{- \alpha L}.
\]
On the other hand, choosing $V= V_{L,0}$ with $L \ge L^\gap$ in Proposition~\ref{prop:expCV}, and denoting by $\gamma_L' \in \cP_\per$ the one-body density matrix defined in~\eqref{eq:gamma_lemma} for this choice of potential, we get
\[
	\forall L \ge L^\gap, \quad I_\per^{\mu_\per} = \cE_\per^{\mu_\per}(\gamma_0) \le \cE_\per^{\mu_\per}(\gamma_L') \le L^{-3} \cE_L^{\mu_\per}(\gamma_{L,0}) + C \re^{- \alpha L} = L^{-3} I_L^{\mu_\per} + C \re^{-\alpha L}.
\]
Combining both inequalities leads to
\begin{equation} \label{eq:rHF_conv_energy}
	\forall L \ge L^\gap, \quad \left| L^{-3} I_L^{\mu_\per} - I_\per^{\mu_\per} \right| \le C \re^{- \alpha L}.
\end{equation}
This leads to the claimed rate of convergence for the ground-state energy per unit cell.

 
 \subsection*{Step 2: \textit{Convergence of the ground state density}} 
 
 In order to compare $\rho_{\gamma_0}$ and $\rho_{\gamma_{L,0}}$, it is useful to introduce the Hamiltonian $H^L := -\dfrac12 \Delta_L + V_0$ acting on $L^2_\per(\WS_L)$. We also introduce $\gamma_L := \mathds{1} (H^L < \varepsilon_F)$. Note that $\gamma_L \in \cP_L$ is the operator obtained in~\eqref{eq:gamma_lemma} by taking $V = V_0$ in Proposition~\ref{prop:expCV}. Therefore, according to this proposition, there exist $C \in \R^+$ and $\alpha > 0$ such that
 \begin{equation} \label{eq:rho0_rhoL}
 	\forall L \in \N^*, \quad \| \rho_{\gamma_0} - \rho_{\gamma_L} \|_{L^\infty(\WS)} \le C \re^{- \alpha L}.
 \end{equation}
In order to compare $\rho_{\gamma_L}$ with $\rho_{\gamma_{L,0}}$, we note that, since $\gamma_{L,0}$ is a minimizer of~\eqref{eq:finite_problem}, then, using~\eqref{eq:conv_energy} and~\eqref{eq:rHF_conv_energy}, we get that, for any $L \in \N^*$,
\begin{align*}
	0  \le L^{-3} \cE_L^{\mu_\per} (\gamma_L) - L^{-3} \cE_L^{\mu_\per}(\gamma_{L,0}) & = \left( L^{-3} \cE_L^{\mu_\per} (\gamma_L)  - \cE_\per^{\mu_\per}(\gamma_0) \right) + \left( \cE_\per^{\mu_\per}(\gamma_0) - L^{-3} \cE_L^{\mu_\per}(\gamma_{L,0}) \right) \\
		& \le 2 C \re^{- \alpha L},
\end{align*}
 so that 
 \[
 	\forall L \in \N^*, \quad 0 \le  \cE_L^{\mu_\per} (\gamma_L) - \cE_L^{\mu_\per}(\gamma_{L,0}) \le L^3 2 C \re^{-\alpha L} \le C' \re^{- \alpha' L}
\]
for some constants $C' \in \R^+$ and $\alpha' > 0$ independent of $L$. This inequality can be recast into
 \[
 	\forall L \in \N^*, \quad 0 \le \VTr_L \left( \left( H_{L,0} - \varepsilon_F \right) (\gamma_L - \gamma_{L,0}) \right) + \frac12 D_1(\rho_{\gamma_L} - \rho_{\gamma_{L,0}}, \rho_{\gamma_L} - \rho_{\gamma_{L,0}}) \le C' \re^{- \alpha' L}.
 \]
 Both terms are non-negative, so each one of them is decaying exponentially fast. From the inequality (recall that we assumed $\left| \ba_1^\ast \right| \le \left| \ba_2^\ast \right| \le \left| \ba_3^\ast \right|$)
 \[
 	\forall f \in L^2_\per(\WS), \quad
	\left\| f \ast_{\WS} G_1 \right\|_{L^2_\per(\WS)}^2  = \sum_{\bR \in \Lat \setminus \{ \bnull \}} \dfrac{\left| c_\bk\left( f \right) \right|^2}{\left| \bk \right|^4}
	\le \dfrac{1}{\left| \ba_3^\ast \right|^2}\sum_{\bR \in \Lat \setminus \{ \bnull \}} \dfrac{\left| c_\bk\left( f \right) \right|^2}{\left| \bk \right|^2}
	= \dfrac{1}{\left| \ba_3^\ast \right|^2} D_1(f, f),
 \]
 we obtain that
 \begin{equation} \label{eq:control_G1_L2}
 	\forall L \in \N^*, \quad \| (\rho_{\gamma_L} - \rho_{\gamma_{L,0}}) \ast_\Gamma G_1 \|_{L^{2}_\per(\WS)}^2 \le \dfrac{1}{\left| \ba_3^\ast \right|^2} D_1(\rho_{L} - \rho_{\gamma_{L,0}}, \rho_{L} - \rho_{\gamma_{L,0}}) \le \dfrac{2 C'}{\left| \ba_3^\ast \right|^2} \re^{- \alpha'L}.
 \end{equation}

Consider $W \in L^{2}_\per(\WS)$. It holds that, for any $L \ge L^\gap$,
 \begin{align}
   	  & \int_{\WS} (\rho_{\gamma_L} - \rho_{\gamma_{L,0}} ) W 
	   = \dfrac{1}{L^3} \sum_{\bQ \in \Lambda_L} \Tr_{L^2_\per(\WS)} \left[  \left( ( \gamma_L)_\bQ - ( \gamma_{L,0} )_\bQ \right) W \right] \nonumber \\
   	  & \quad = \dfrac{1}{2 \ri \pi L^3} \sum_{\bQ \in \Lambda_L} \oint_{\sC}\Tr_{L^2_\per(\WS)} \left( \dfrac{1}{\lambda - \left( H_0 \right)_\bQ} \left( (\rho_{\gamma_{L,0}} - \rho_{\gamma_L}) \ast_\WS G_1 \right) \dfrac{1}{\lambda - \left( H_{L,0} \right)_\bQ} W\right) \rd \lambda \label{eq:intermediate_calculation}\\
	  & \quad = \dfrac{1}{2 \ri \pi L^3} \sum_{\bQ \in \Lambda_L} \oint_{\sC}\Tr_{L^2_\per(\WS)} \left( B_2^\per(\lambda, \bQ) \dfrac{1}{1 - \Delta_1} \left( (\rho_{\gamma_{L,0}} - \rho_{\gamma_L}) \ast_\WS G_1 \right) B_2^L(\lambda, \bQ) \dfrac{1}{1 - \Delta_1} W\right) \rd \lambda, \nonumber
\end{align}
where $B_2^\per$ is the operator defined in~\eqref{eq:def:B1_B2} for $H = H_0$, and $B_2^L$ is the one for $H = H_{L,0}$. From the expression of the constant $C_1$ in~\eqref{eq:def_C}, we deduce that there exists a constant $\widetilde{C_1} \in \R^+$ such that, for all $L \ge L^\gap$,
\[
	\forall \lambda \in \sC, \quad
	\forall \bQ \in \Lambda_L, \quad
	\left\| B_2^\per(\lambda, \bQ ) \right\|_{\cB(L^2_\per(\WS))} \le \widetilde{C_1}
	\quad \text{and} \quad
	\left\| B_2^L(\lambda, \bQ ) \right\|_{\cB(L^2_\per(\WS))} \le \widetilde{C_1}.
\]
As a result,
\[
	\left| \int_{\WS} (\rho_{\gamma_L} - \rho_{\gamma_{L,0}} ) W  \right| \le \dfrac{\left| \sC \right| \widetilde{C_1}^2}{2 \pi}
	\left\| \dfrac{1}{1 - \Delta_1} \left( (\rho_{\gamma_{L,0}} - \rho_{\gamma_L}) \ast_\WS G_1 \right) \right\|_{\fS_2(L^2_\per(\WS))} 
	\left\|  \dfrac{1}{1 - \Delta_1} W \right\|_{\fS_2(L^2_\per(\WS))}.
\]
We deduce from the Kato-Seiler-Simon inequality~\cite[Theorem 4.1]{Simon2005} and the estimate~\eqref{eq:control_G1_L2} that there exists constant $C \in \R^+$ and $\alpha > 0$ independent of $W$ such that,
\[
	\left| \int_{\WS} (\rho_{\gamma_L} - \rho_{\gamma_{L,0}} ) W  \right| \le C \re^{-\alpha L} \left\| W \right\|_{L^2_\per(\WS)}.
\]
This being true for all $W \in L^2_\per(\WS)$, we obtain
\[
	\forall L \ge L^\gap, \quad \left\| \rho_{\gamma_L} - \rho_{\gamma_{L,0}} \right\|_{L^2_\per(\WS)} \le C \re^{-\alpha L}.
\]
	   
This proves the convergence in $L^2_\per(\WS)$. To get the convergence in $L^\infty_\per(\WS)$, we bootstrap the procedure. Since $(\rho_{\gamma_{L,0}} - \rho_{\gamma_L}) \in L^2_\per(\WS)$, then $(\rho_{\gamma_{L,0}} - \rho_{\gamma_L}) \ast_\WS G_1 \in L^\infty_\per(\WS)$ with
\begin{equation} \label{eq:ineq_exp_G1}
	\forall L \ge L^\gap, \quad \left\| 	(\rho_{\gamma_{L,0}} - \rho_{\gamma_L}) \ast_\WS G_1 \right\|_{L^\infty_\per}(\WS) \le C' \re^{-\alpha L}.
\end{equation}
Consider $W\in L^1_\per(\Gamma)$. Performing similar calculations as in~\eqref{eq:intermediate_calculation}, we get (with obvious notation)
 \begin{align*}
   	&\int_{\WS} (\rho_{\gamma_L} - \rho_{\gamma_{L,0}} ) W \\
	&\quad
	= \dfrac{1}{2 \ri \pi L^3} \sum_{\bQ \in \Lambda_L} \oint_{\sC} \Tr_{L^2_\per(\WS)} \left( 
		B_1^\per(\lambda, \bQ) \left( (\rho_{\gamma_{L,0}} - \rho_{\gamma_L}) \ast_\WS G_1 \right) B_2^L(\lambda, \bQ) \dfrac{1}{1 - \Delta_1 }W \dfrac{1}{1 - \Delta_1} \right)\rd \lambda,
\end{align*}
so that
 \begin{align*}
   	\left| \int_{\WS} (\rho_{\gamma_L} - \rho_{\gamma_{L,0}} ) W \right| \le \dfrac{\left| \sC \right| \widetilde{ C_1}^2}{2 \pi}
	\left\|  (\rho_{\gamma_{L,0}} - \rho_{\gamma_L}) \ast_\WS G_1\right\|_{L^\infty_\per(\WS)} 
	\left\| \dfrac{1}{1 - \Delta_1 }W \dfrac{1}{1 - \Delta_1} \right\|_{\fS_1(L^2_\per(\WS))},
\end{align*}
 and we conclude from~\eqref{eq:KSS_W_L1} and~\eqref{eq:ineq_exp_G1} that there exist constants $C \in \R^+$ and $\alpha > 0$ such that
\[
 	\forall L \ge L^\gap, \quad \left\| \rho_{\gamma_L} - \rho_{\gamma_{L,0}} \right\|_{L^\infty_\per(\WS)} \le C \re^{- \alpha L}.
 \]
 Together with~\eqref{eq:rho0_rhoL}, we finally obtain
 \[
 	\forall L \ge L^\gap, \quad \left\| \rho_{\gamma_0} - \rho_{\gamma_{L,0}} \right\|_{L^\infty_\per(\WS)} \le C \re^{- \alpha L}.
 \]
 
 \subsection*{Step 3: \textit{Convergence of the mean-field Hamiltonian}}
 Finally, since
 \[
 	H_{L} - H_0= (\rho_{\gamma_{L,0}} - \rho_{\gamma_0}) \ast_{\WS} G_1,
 \]
the estimate~\eqref{eq:ineq_exp_G1} implies the convergence of the operator $H_L - H_0$ to $0$ in $\cB(L^2(\R^3))$ with an exponential rate of convergence. 

 \begin{remark}
  The convergence of the operators implies the convergence of the eigenvalues. More specifically, from the min-max principle, we easily deduce that
  \[
			\sup_{\bq \in \BZ} \sup_{n \in \N^\ast} \left| \varepsilon_{n, \bq}[H_{L}] - \varepsilon_{n, \bq} [H_0] \right| \le C \re^{- \alpha L}
	\]
where $\left( \varepsilon_{n,\bq}[H] \right)_{n \in \N^*}$ denotes the eigenvalues of the operator $H_\bq$ ranked in increasing order, counting multiplicities.
 \end{remark}

%
%

\section{Numerical simulations}
\label{sec:numerics_rHF}

In this final section, we illustrate our theoretical results with numerical simulations. The simulations were performed using a home-made Python code, run on a 32 core Intel Xeon E5-2667.

\paragraph{The linear model (Proposition~\ref{prop:expCV}) } ~\\
We consider crystalline silicon in its diamond structure. A qualitatively correct band diagram of this system can be obtained from a linear Hamiltonian of the form $H = -\frac12 \Delta + V_\per^{\rm lin}$, where the potential $V_\per^{\rm lin}$ is the empirical pseudopotential constructed in~\cite{Cohen1966}. The lattice vectors are 
\[
 	\ba_1 = \frac{a}{2} (0, 1, 1)^T, \quad \ba_2 = \frac{a}{2} (1, 0, 1)^T \quad \text{and} \quad \ba_3 = \frac{a}{2} (1, 1, 0)^T
\]
and the reciprocal lattice vectors are
\[
	\ba_1^\ast = \frac{2 \pi}{a} (-1, 1, 1)^T, \quad \ba_2^\ast = \frac{2 \pi}{a} (1, -1, 1)^T \quad \text{and} \quad \ba_3^\ast = \frac{2 \pi}{a} (1, 1, -1)^T,
\]
where the lattice constant is~\cite{Cohen1966} $a = 10.245~\rm{Bohr}$ (that is about $a = 5.43~\textrm{\AA}$). 
The pseudopotential $V_\per^{\rm lin}$ is given by the expression~\cite{Cohen1966}
\begin{equation} \label{eq:Vper_lin}
	V_\per^{\rm lin}(\bx) = \sum_{\bk \in \RLat} V_\bk \re^{ \ri \bk \cdot \bx} \quad \text{with} \quad
	\forall \bk \in \RLat, \quad V_\bk = S[\bk] \cos \left( \frac{a(k_1 + k_2 + k_3)}{8} \right)
\end{equation}
where
\[
	S[\bk] = \left\{ \begin{array}{lll}
		-0.105 &\text{if} & | \bk |^2 = 3 (2\pi/a)^2 \\
		0.02 & \text{if} & | \bk |^2 = 8 (2 \pi /a)^2 \\
		0.04 & \text{if} & | \bk |^2 = 11 ( 2 \pi /a)^2 \\
		0 & \text{otherwise}. &
		\end{array} \right.
\]

%
%
%

This system is an insulator when the number of particle (electron-pairs) $N$ per unit cell is~$N = 4$, so that the hypotheses of Proposition~\ref{prop:expCV} are satisfied. In the sequel, the calculations are performed in the planewave basis
\begin{equation} \label{eq:def:X}
	X = \left\{ e_\bk, \ \bk \in \RLat, \ \frac{\left| \bk \right|^2}{2} < E_\cutoff \right\},
\end{equation}
where the cut-off energy is $E_\cutoff = 736 \ \rm{eV}$. The corresponding size of the basis is $| X | = 749$.  \\

In Figure~\ref{fig:linear_cvExp}, we represent the error on the ground state energy per unit cell and the $L^\infty(\R^3)$ error on the ground state density (in log scale) for different sizes of the regular grid. The value of $L$ in~\eqref{eq:def:LambdaL} varies between $4$ to $28$. The quantities of reference are the ones calculated for the regular grid of size $60$. We observe in Figure~\ref{fig:linear_cvExp} the exponential convergence for both the energy per unit cell and the density as predicted in Proposition~\ref{prop:expCV}.

\begin{figure}[h!]
\begin{center}
\includegraphics[scale = 0.6]{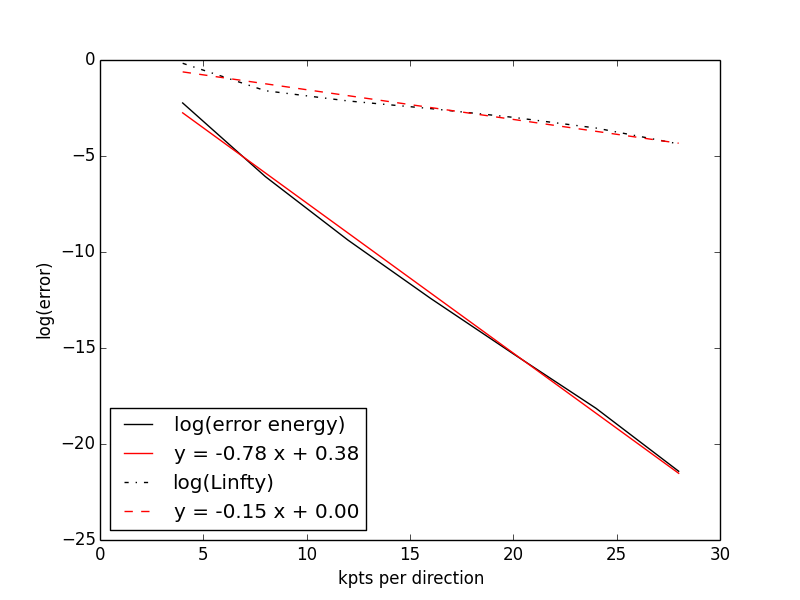}

\caption{The error on the ground-state energy (in eV) and the $L^\infty$ error on the ground-state density with respect to the size of the regular mesh for the linear model. The logarithm of the errors are represented. The linear regression curves are also displayed.}
\label{fig:linear_cvExp}

\end{center}
\end{figure}

\paragraph{The rHF model (Theorem~\ref{th:Deleurence_exp})} ~\\
We now consider the rHF model. To our knowledge, no pseudopotential has ever been designed for this model. Since constructing pseudopotentials is a formidable task, we limit ourselves to the following poor man's solution, which does not aim at capturing the physics but only at illustrating numerically our theoretical convergence results. We decompose the potential self-consistent $V_0$ appearing in~\eqref{eq:sc} into
\[
	V_0 = \left( \rho_{\gamma_0} - \mu_\per \right) \ast_\WS G_1 = \rho_{\gamma_0} \ast_\WS G_1 - \mu_\per \ast_\WS G_1,
\]
and we make the approximation $V_0 = V_\per^{\rm lin}$, where $V_\per^{\rm lin}$ is the pseudopotential defined in~\eqref{eq:Vper_lin}. This leads to the rHF pseudopotential of the form
\[
	V_\per^{\rm HF} := V_\per^{\rm lin} - \rho_{\gamma_0} \ast_\WS G_1.
\]
In practice, we calculate $V_\per^{\rm HF}$ with the potential $\rho_{\gamma_0}$ obtained previously for the grid of size~$60$. The minimization problem~\eqref{eq:finite_problem}-\eqref{eq:sc_supercell} is solved self-consistently in the basis $X$ defined in~\eqref{eq:def:X} (we refer to~\cite{Cances2000} for a survey on self-consistent procedures for such problems). We stop the self-consistent procedure when the $L^\infty(\R^3)$ difference between two consecutive densities is less than~$10^{-7}$. The size of the regular mesh varies between $8$ to $36$. The quantities of reference are the ones calculated for the regular mesh of size $60$. The error on the energy per unit cell and the $L^\infty(\R^3)$ error on the density are displayed in Figure~\ref{fig:rHF_cvExp}.

\begin{figure}[h!]
\begin{center}
\includegraphics[scale = 0.6]{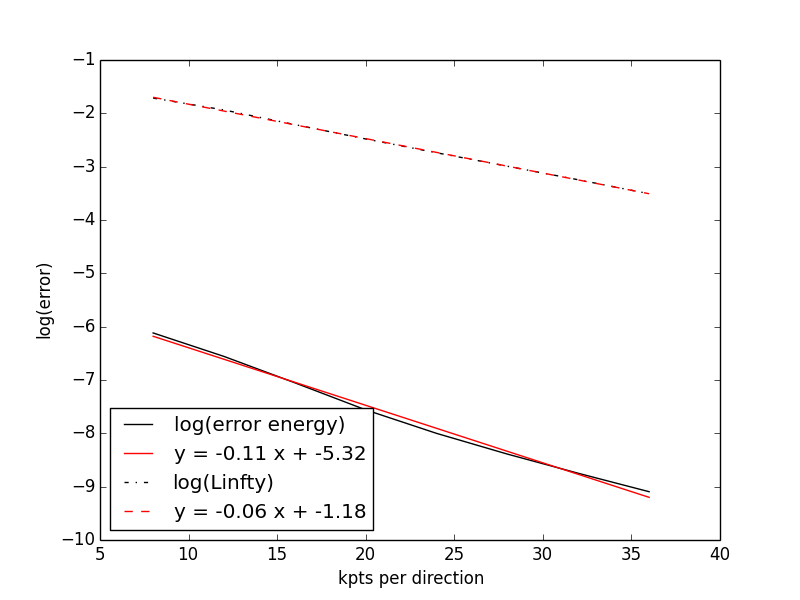}

\caption{The error in energy (in eV) and the $L^\infty$ error of the density with respect to the size of the regular mesh for the rHF model. The logarithm of the errors are represented. The linear regression curves are also displayed.}
\label{fig:rHF_cvExp}

\end{center}
\end{figure}

We observe in Figure~\ref{fig:rHF_cvExp} the exponential convergence announced in Theorem~\ref{th:Deleurence_exp}. 

\section*{Acknowledgements}

We are grateful to Gianluca Panati and Domenico Monaco for answering our questions about analyticity. We also thank Eric Cancès for his suggestions and help to improve the paper.

\bibliographystyle{plain}

\bibliography{perfect_crystal_rHF}

\end{document}

%% file: usepackages.tex
\usepackage{a4wide}
\usepackage{amsmath,amssymb,amsfonts,amsthm}
\usepackage{moreverb,rotating,graphics}
\usepackage[hidelinks]{hyperref}
\usepackage[utf8]{inputenc}
\usepackage{epsfig}
\usepackage[english]{babel} 
\usepackage{framed,fancybox}
\usepackage{tikz,pgfplots}
\usepackage{caption}
\usepackage{subcaption}
\usepackage{mathrsfs}
\usepackage{esint}
\usepackage{dsfont}


%% file: list_commands.tex
\numberwithin{equation}{section} 

\newtheorem{theorem}{Theorem}[section]
\newtheorem{proposition}[theorem]{Proposition}
\newtheorem{remark}[theorem]{Remark}
\newtheorem{lemma}[theorem]{Lemma}
\newtheorem{corollary}[theorem]{Corollary}

\newcommand\1{{\ensuremath {\mathds 1} }} 

\def\C{{\mathbb C}}

\def\N{{\mathbb N}}

\def\R{{\mathbb R}}
\def\RR{{\mathbb R}}

\def\Z{{\mathbb Z}}
\def\ZZ{{\mathbb Z}}

\def\ba {{ \bold a}}

\def\bk{{\bold k}}

\def\bm{{\bold m}}

\def\bQ{{\bold Q}}
\def\bq{{\bold q}}
\def\bR{{\bold R}}

\def\bx{{\bold x}}
\def\by{{\bold y}}
\def\bz{{\bold z}}

\def\bnull{{\bold 0}}

\def\rd{{\mathrm{d}}}
\def\re{{\mathrm{e}}}
\def\ri{{\mathrm{i}}}

\def\Re{{\mathrm{Re }\, }}
\def\Im{{\mathrm{Im }\, }}

\def\Tr{{\rm Tr}}
\def\VTr{ \underline{ \rm Tr \,}}

\def\cB{{\mathcal B}}

\def\cD{{\mathcal D}}
\def\cE{{\mathcal E}}

\def\cH{{\mathcal H}}

\def\cP{{\mathcal P}}

\def\cS{{\mathcal S}}

\def\cZ{{\mathcal Z}}

\def\fS{{\mathfrak S}}

\newcommand{\sC}{\mathscr{C}}

\newcommand{\bra}{\langle}
\newcommand{\ket}{\rangle}

\newcommand{\BZ}{{\Gamma^\ast}} 
\newcommand{\WS}{\Gamma} 
\newcommand{\RLat}{\mathcal{R}^*} 
\newcommand{\Lat}{\mathcal{R}} 

\newcommand{\per}{\mathrm{per}}

\newcommand{\loc}{\mathrm{loc}}

\newcommand{\gap}{\mathrm{gap}}

\def\sqw{\hbox{\rlap{\leavevmode\raise.3ex\hbox{$\sqcap$}}$%
\sqcup$}}
\def\cqfd{\ifmmode\sqw\else{\ifhmode\unskip\fi\nobreak\hfil
\penalty50\hskip1em\null\nobreak\hfil\sqw
\parfillskip=0pt\finalhyphendemerits=0\endgraf}\fi}

\hypersetup{
    colorlinks,
    linkcolor={blue!50!blue},
    citecolor={red}
}
\renewcommand{\eqref}[1]{(\ref{#1})}


%% file: perfect_crystal_rHF_v9.bbl
\begin{thebibliography}{10}

\bibitem{Blanc2003}
X.~Blanc, C.~{Le Bris}, and P.-L. Lions.
\newblock A definition of the ground state energy for systems composed of
  infinitely many particles.
\newblock {\em Comm. Part. Diff. Eq.}, 28:439--475, 2003.

\bibitem{Brouder2007}
C.~Brouder, C.~Panati, M.~Calandra, C.~Mourougane, and N.~Marzari.
\newblock Exponential localization of {W}annier functions in insulators.
\newblock {\em Phys. Rev. Lett.}, 98(4):046402, 2007.

\bibitem{Cances2000}
E.~Canc{\`e}s.
\newblock {SCF algorithms for Hartree-Fock electronic calculations}.
\newblock In M.~Defranceschi and C.~{Le~Bris}, editors, {\em Mathematical
  Models and Methods for Ab Initio Quantum Chemistry}. Springer, 2000.

\bibitem{Cances2008}
E.~Canc\`es, A.~Deleurence, and M.~Lewin.
\newblock {A new approach to the modeling of local defects in crystals: the
  reduced Hartree-Fock case}.
\newblock {\em Commun. Math. Phys.}, 281:129--177, 2008.

\bibitem{Catto2001}
I.~Catto, C.~{Le Bris}, and P.-L. Lions.
\newblock {On the thermodynamic limit for Hartree-Fock type models}.
\newblock {\em Ann. I. H. Poincaré (C)}, 18(6):687--760, 2001.

\bibitem{Cohen1966}
M.L. Cohen and T.K. Bergstresser.
\newblock Band structures and pseudopotential form factors for fourteen
  semiconductors of the diamond and {Z}inc-blende structures.
\newblock {\em Phys. Rev.}, 141:789--796, 1966.

\bibitem{Cloizeaux1964}
J.~Des~Cloizeaux.
\newblock Analytical properties of $n$-dimensional energy bands and {W}annier
  functions.
\newblock {\em Phys. Rev.}, 135:A698--A707, 1964.

\bibitem{Cloizeaux1964a}
J.~Des~Cloizeaux.
\newblock Energy bands and projection operators in a crystal: Analytic and
  asymptotic properties.
\newblock {\em Phys. Rev.}, 135:A685--A697, 1964.

\bibitem{Kato2012}
T.~Kato.
\newblock {\em {Perturbation Theory for Linear Operators}}.
\newblock Springer Science \& Business Media, 2012.

\bibitem{Kohn1959}
W.~Kohn.
\newblock Analytic properties of {B}loch waves and {W}annier functions.
\newblock {\em Phys. Rev.}, 115:809--821, 1959.

\bibitem{Lieb1977}
E.H. Lieb and B.~Simon.
\newblock {The Thomas-Fermi theory of atoms, molecules and solids}.
\newblock {\em Adv. Math.}, 23:22--116, 1977.

\bibitem{Monkhorst1976}
H.J. Monkhorst and J.D. Pack.
\newblock {Special points for Brillouin-zone integrations}.
\newblock {\em Phys. Rev. B}, 13(12):5188--5192, 1976.

\bibitem{Panati2007}
G.~Panati.
\newblock Triviality of {B}loch and {B}loch--{D}irac bundles.
\newblock {\em Ann. H. Poincar{\'e}}, 8(5):995--1011, 2007.

\bibitem{ReedSimon4}
M.~Reed and B.~Simon.
\newblock {\em {Methods of Modern Mathematical Physics. Analysis of
  Operators}}, volume~IV.
\newblock Academic Press, 1978.

\bibitem{Simon2005}
B.~Simon.
\newblock {\em {Trace Ideals and Their Applications}}.
\newblock Mathematical Surveys and Monographs. American Mathematical Society,
  2005.

\end{thebibliography}
